\documentclass{amsart}%
\usepackage{amsmath}
\usepackage{amsfonts}
\usepackage{amssymb}
\usepackage{graphicx}
\usepackage{todonotes}
\usepackage{mathtools}
\usepackage[linesnumbered,lined,boxed,commentsnumbered,rightnl]{algorithm2e}
\usepackage{caption}
\usepackage{subcaption}%
\providecommand{\U}[1]{\protect\rule{.1in}{.1in}}
\newtheorem{theorem}{Theorem}[section]

\newtheorem{assumption}[theorem]{Assumption}

\newtheorem{example}[theorem]{Example}

\newtheorem{lemma}[theorem]{Lemma}

\newtheorem{proposition}[theorem]{Proposition}
\newtheorem{remark}[theorem]{Remark}

\DeclareMathOperator{\esssup}{ess\,sup}
\DeclareMathOperator{\spn}{span}

\DeclareMathOperator*{\argmin}{arg\,min}
\numberwithin{equation}{section}
\begin{document}
\title[Dyn. programming for opt. stopping via pseudo-regression]{Dynamic programming for optimal stopping via pseudo-regression}
\author{Christian Bayer, Martin Redmann, John Schoenmakers}
\maketitle

\begin{abstract}
We introduce new variants of classical regression-based algorithms for optimal
stopping problems based on computation of regression coefficients by Monte
Carlo approximation of the corresponding $L^{2}$ inner products instead of the
least-squares error functional. Coupled with new proposals for simulation of
the underlying samples, we call the approach ``pseudo regression''. A detailed convergence analysis is provided and it is shown that
the approach asymptotically leads to less computational cost for a
pre-specified error tolerance, hence to lower complexity. The method is
justified by numerical examples.

\end{abstract}

\section{Introduction}

Stochastic optimal stopping problems (in discrete time) play an important role
in the theoretical as well as in the numerical literature on stochastic
optimal control, since they are both generally considered difficult to solve
and have many practical applications, in particular in energy and finance
(where American or Bermudan options can naturally be understood as stochastic
optimal stopping problems).

Many numerical methods have been suggested, ranging from PDE techniques (based
on the Hamilton-Jacobi-Bellman equation of the associated continuous-time
problem), to Monte Carlo (simulation) based approaches involving regression
techniques, policy iteration, duality, and more. For an overview, see for
instance \cite{G04}, \cite{BS18}.

In this paper, we consider stochastic approaches based on the Bellman
equation. A key ingredient of the classical algorithms such as the ones
proposed by Longstaff and Schwartz~\cite{LS2001} or Tsitsiklis and Van
Roy~\cite{TV2001} is (global) regression, used to compute a conditional
expectation of, say, $u(z) \coloneqq \mathbb{E}[Y|Z=z]$ for some random
variables $Y$ and $Z$. Given basis functions $\psi_{1}, \ldots, \psi_{K}$, one
thus looks for the best approximation of the unknown function $u$ in the
linear $\mathrm{span} \{\psi_{1}, \ldots, \psi_{K}\}$ with respect to the
distribution of $Z$ denoted by $\mu$, i.e., we would ideally like to solve the
minimization problem
\[
\beta^{\ast}\coloneqq \argmin_{\beta\in\mathbb{R}^{K}} \mathbb{E}\left[
\left|  Y - \sum_{k=1}^{K} \beta_{k} \psi_{k}(Z) \right|  ^{2} \right]  ,
\]
in order to find an approximation $u(\cdot) \approx\sum_{k=1}^{K} \beta
_{k}^{\ast}\psi_{k}(\cdot) \eqqcolon u^{K}(\cdot)$. Classically, the above
minimization problem is directly translated into the corresponding
least-squares problem based on Monte Carlo approximation of the expectation,
i.e., for i.i.d.~samples $(Y^{i},Z^{i})$, $i=1, \ldots, M$, one solves
\begin{equation}
\label{eq:least-squares-regression}\widehat{\beta} \coloneqq \argmin_{\beta
\in\mathbb{R}^{K}} \sum_{i=1}^{M} \left|  Y^{i} - \sum_{k=1}^{K} \beta_{k}
\psi_{k}(Z^{i}) \right|  ^{2}.
\end{equation}
While well-understood by now, it is worth-while to recall that the analysis of
the convergence of $\widehat{\beta}$ as $M \to\infty$ is not trivial due to
the reliance on random matrix theory, see, for instance, \cite{Gy2002}.
Instead of approximating the minimization problem by Monte Carlo simulation it
is also possible to directly approximate the solution $\beta^{\ast}$. Indeed,
note that $u^{K}$ is, of course, the linear projection of $u$ to
$\mathrm{span} \{\psi_{1}, \ldots, \psi_{K}\}$ in the $L^{2}(\mu)$-sense.
Hence, assuming for ease of notation that the basis functions $\psi_{1},
\ldots, \psi_{K}$ are orthonormal w.r.t.~$\mu$---the general case requires
multiplication with the Gram matrix formed by $\left\langle \psi_{k} \, ,
\psi_{l} \right\rangle _{L^{2}(\mu)}$---we have
\[
\beta_{k}^{\ast}= \left\langle u \, , \psi_{k} \right\rangle _{L^{2}(\mu)} =
\mathbb{E}\left[  \mathbb{E}[Y|Z] \psi_{k}(Z) \right]  = \mathbb{E}[ Y
\psi_{k}(Z) ].
\]
This formula, however, can be immediately approximated by Monte Carlo
simulation giving
\begin{equation}
\label{eq:L2-regression}\overline{\beta}_{k} \coloneqq \frac{1}{M} \sum
_{i=1}^{M} Y^{i} \psi_{k}(Z^{i}), \quad k=1, \ldots, K.
\end{equation}
From a technical point of view, convergence analysis of $\overline{\beta}$ is
relatively straightforward and leads to squared error terms of the order
$\frac{K}{M}$ (see Theorem~\ref{psth}). On the other hand, the squared error
due to the solution $\widehat{\beta}$ of the least squares problem is of order
$\frac{(1+\ln M) K}{M}$ (see Theorem~\ref{rgth}). At the same time, computing
$\overline{\beta}$ is also cheaper compared to computing $\widehat{\beta}$, as
we avoid computing a system of linear equations (see the discussions in
Section~\ref{sec:computational-cost}). However, as we see in the later
Section~\ref{alg}, computation of $\overline{\beta}$ does rely on knowledge of
the Gram matrix associated to the basis functions and the measure $\mu$.

Another important detail of regression based algorithms, especially as
consecutive regression steps are required, is the choice of random variables
$(Y,Z)$. Clearly, the result of the regression procedure (just as the
conditional expectation) only depends on the conditional distribution of $Y$
given $Z$, but not on the distribution of $Z$ itself, which gives us
considerable freedom.

In the context of Bermudan options, let $X_{j}$ denote the underlying process
at time $j$ (see Section~\ref{sec:recap-optim-stopp} for more details) and let
$v_{j}$ denote the option value at time $j$. Then, in the simplest case of
dynamical programming, we need to evaluate conditional expectations
$\mathbb{E}[v_{j}(X_{j})|X_{j-1} = z]$. Hence, a very natural implementation
of the regression procedure above will be based on $M$ samples $(X^{i}_{0},
\ldots, X^{i}_{\mathcal{J}})$ of the whole trajectory until the expiry time
$\mathcal{J}$ of the option, iteratively using slices $Y^{i} \equiv
v_{j}(X^{i}_{j})$ and $X^{i} \equiv X^{i}_{j-1}$ in the above notation, for
$j=1, \ldots, \mathcal{J}$. Hence, the distribution $\mu$ of $X$ will depend
on $j$.

An alternative approach, especially advantageous when $X_{j}$ is a homogeneous
Markov process, i.e., when the conditional distribution of $X_{j}$ given
$X_{j-1} = z$ does not depend on $j$, is to fix a (carefully chosen)
probability measure $\mu$ for all $j$. Now sample r.v.s $U^{i}$ from $\mu$ and
$X^{i}$ from the conditional distribution of $X_{j}$ given $X_{j-1} = U^{i}$.
Hence, we obtain
\[
\mathbb{E}[v_{j}(X)|U=z] = \mathbb{E}[v_{j}(X_{j})|X_{j-1} = z],
\]
and we can use the same batch of samples for each of the consecutive
regression steps for $j=1, \ldots, \mathcal{J}$, considerably reducing the
computational time of the algorithm. As an added benefit, we are now free to
choose the probability measure $\mu$. This allows us to specifically choose
both $\mu$ and the basis function $\psi_{1}, \ldots, \psi_{K}$ such that the
basis functions are already orthogonal w.r.t.~$\mu$, implying a trivial Gram matrix.

In what follows, we call the combination of using a fixed set of sampled
trajectories $X^{i}_{0}, \ldots, X^{i}_{\mathcal{J}}$ with the least-squares
estimator~(\ref{eq:least-squares-regression}) \emph{standard regression}, and
we call a combination of samples $(U^{i},X^{i})$ based on a arbitrarily chosen
measure $\mu$ together with the $L^{2}$-projection
estimator~(\ref{eq:L2-regression}) \emph{pseudo regression}. We argue that
pseudo regression has both theoretical and numerical advantages compared with
standard regression for many Bermudan option problems. Indeed,

\begin{itemize}
\item the convergence rates for the number of samples $M \to\infty$ are better
due to the missing $\ln(M)$-term (see Theorems~\ref{psth} and~\ref{rgth});

\item the asymptotic number of floating point operations necessary is smaller
(see Section~\ref{sec:computational-cost});

\item numerical examples indicate lower computational costs for fixed error
tolerance, in line with the theory, see Section~\ref{sec:numer-exper}.
\end{itemize}
Last but not least, we provide a detailed 
analysis yielding explicit convergence rates for the pseudo regression versions of both the Tsitsiklis--van Roy
and the Longstaff--Schwartz algorithm.

\subsection*{Outline of the paper}

In Section~\ref{sec:recap-optim-stopp} we recapitulate some theory of optimal
stopping in discrete time and recall the (classical) Tsitsiklis--van Roy and
Longstaff--Schwartz algorithms. In Section~\ref{Mark} we
describe in detail the one-step regression procedures involved for both
\emph{standard} and \emph{pseudo regression}. In
Section~\ref{gvps} we state a general convergence result for the
pseudo-regression approach (Theorem~\ref{psth}), a convergence result for the
pseudo-regression version of Longstaff--Schwartz (Theorem~\ref{thm_main}), and
a similar convergence result for the pseudo-regression version of
Tsitsiklis--van Roy (Theorem~\ref{thm_mainTV}). We discuss the computational
cost for the different variants of the algorithms in
Section~\ref{sec:computational-cost} and give numerical examples in
Section~\ref{sec:numer-exper}. We conclude with a summary and an outline of
future research in Section~\ref{sec:conclusions}. More technical
proofs are deferred to the Appendix section.

\section{Recap of optimal stopping in discrete time}

\label{sec:recap-optim-stopp}

\subsection{Theory of optimal stopping in discrete time}

\label{sec:theory-optim-stopp}

Let us recall some facts about the optimal stopping problem in discrete time.
Suppose $(Z_{j}$\thinspace$:$ $j=0,1,\ldots,\mathcal{J})$ is a nonnegative
adapted stochastic process in discrete time on a filtered probability space
$(\Omega,\mathcal{F}_{j},0\leq j\leq\mathcal{J},P),$ which satisfies
\[
\sum_{j=1}^{\mathcal{J}}\mathbb{E}\left[  Z_{j}\right]  <\infty.
\]
In the context of a (discrete time) American or Bermudan option $Z$ may be
regarded as a (discounted) cash-flow process that may be exercised once by the
option holder. More specifically, one may think of $P$ as a pricing measure
corresponding to some num\'{e}raire $\mathcal{N}$ (with $\mathcal{N}_{0}=1$
for simplicity), and $Z=R/\mathcal{N},$ where $(R_{j}$\thinspace$:$
$j=0,1,\ldots,\mathcal{J})$ is a real (not discounted) cash-flow process.
Then, from general no arbitrage principles it is well known that a fair price
of the American option is given by%
\begin{equation}
Y_{0}:=\sup_{\tau\in\mathcal{S}_{0}}\mathbb{E}\left[  Z_{\tau}\right]  ,
\label{os}%
\end{equation}
where $\mathcal{S}_{0}$ denotes the set of $\mathcal{F}$-stopping times taking
values in $\{0,\ldots,\mathcal{J}\}.$ The \textit{Snell envelope}\emph{ }of
$Z$ is defined as%
\begin{equation}
Y_{j}:=\esssup_{\tau\in\mathcal{S}_{j}}\mathbb{E}_{\mathcal{F}_{j}}\left[
Z_{\tau}\right]  ,\text{ \ \ }j=0,...,\mathcal{J}, \label{snel}%
\end{equation}
where $\mathcal{S}_{j}$ denotes the set of $\mathcal{F}$-stopping times taking
values in $\{j,\ldots,\mathcal{J}\}.$ We recall the following classical facts
(e.g. see \cite{Nev75}):

\begin{enumerate}
\item The Snell envelope $Y$ of $Z$ is the smallest super-martingale that
dominates $Z.$ It can be constructed recursively by the \textit{Backward
Dynamic Program principle} or \textit{Bellman principle}:
\begin{align}
Y_{\mathcal{J}}  &  =Z_{\mathcal{J}}\label{BP}\\
Y_{j}  &  =\max\left(  Z_{j},\mathbb{E}_{\mathcal{F}_{j}}\left[
Y_{j+1}\right]  \right)  ,\text{ \ \ }0\leq j<\mathcal{J}.\nonumber
\end{align}

\item An optimal stopping time for (\ref{os}) is given by
\[
\tau^{\ast}=\min\{j:\;0\leq j\leq\mathcal{J},\;Z_{j}\geq\mathbb{E}%
_{\mathcal{F}_{j}}\left[  Y_{j+1}\right]  ]\}
\]
with $Y_{\mathcal{J}+1}:=0.$ That is,%
\[
Y_{0}=\sup_{\tau\in\mathcal{S}_{0}}\mathbb{E}\left[  Z_{\tau}\right]
=\mathbb{E}\left[  Z_{\tau^{\ast}}\right]  .
\]

\end{enumerate}

Thus, in principle, one may arrive at the solution to (\ref{os}) by carrying
out (\ref{BP}) backwardly from $j=\mathcal{J}$ down to $j=0.$ However,
straightforwardly, this leads to a high degree nested expression of
conditional expectations that is virtually impossible to evaluate in practice.

Let us now assume the presence of an underlying Markovian process
$X:=X^{0,x}:=(X_{j}^{0,x}$\thinspace$:$ $j=0,1,\ldots,\mathcal{J}),$ adapted
to $\left(  \mathcal{F}_{j}\right)  ,$ living in $\mathbb{R}^{d},$ and
starting at $X_{0}^{0,x}=x$ a.s. More generally, $\left(  X_{r}^{j,z}%
:r=j,...,\mathcal{J}\right)  $ denotes a random trajectory with $X_{j}%
^{j,z}=z$ a.s. Let us further assume that the cash-flow has the form%
\[
Z_{j}\left(  \omega\right)  =f_{j}(X_{j}\left(  \omega\right)  ),\ \ 0\leq
j\leq\mathcal{J}%
\]
for some functions $\,f_{j}(\cdot):\mathbb{R}^{d}\rightarrow\mathbb{R}_{\geq
0}.$ Then, due to Markovianity, there exist functions $v_{j}(\cdot
):\mathbb{R}^{d}\rightarrow\mathbb{R}_{\geq0},$ such that we may similarly
write%
\[
Y_{j}\left(  \omega\right)  =v_{j}(X_{j}\left(  \omega\right)  ),\ \ 0\leq
j\leq\mathcal{J}.
\]
The Bellman principle now simply says that
\[
v_{j}(X_{j}) = \max\left(  f_{j}(X_{j}), \mathbb{E}[v_{j+1}(X_{j+1}) | X_{j}]
\right)  ,\qquad j<\mathcal{J}.
\]
Henceforth
\[
c_{j}(x) \coloneqq \mathbb{E}[v_{j+1}(X_{j+1}) | X_{j} = x]
\]
is called the \emph{continuation value function}. The numerically challenging
task is, of course, the computation of the $c_{j}$ for $0$ $\leq$ $j$ $<$
$\mathcal{J}$.

\subsection{Standard regression algorithms}

\label{standardregressionTvRLS}

For clarity, let us describe the classical Tsitsiklis--van Roy algorithm in
full detail. Let $(X_{0}^{(m)},\ldots,X_{\mathcal{J}}^{(m)})$, $m=1,\ldots,M$,
denote $M$ independent trajectories from the Markov process $X$. Initialize
$\widehat{v}_{\mathcal{J}}\coloneqq f_{\mathcal{J}},$ $\widehat{c}%
_{\mathcal{J}}\coloneqq0.$ If $\widehat{v}_{j}$ and $\widehat{c}_{j}$ are
already constructed, iteratively construct (backward in time)
\begin{align}
\widehat{\beta}^{(j-1)}  &  \coloneqq\argmin_{\beta\in\mathbb{R}^{K}}%
\sum_{m=1}^{M}\left(  \widehat{v}_{j}(X_{j}^{(m)})-\sum_{k=1}^{K}\beta_{k}%
\psi_{k}(X_{j-1}^{(m)})\right)  ^{2},\label{eq:TVR-classical-regression}\\
\widehat{c}_{j-1}(\cdot)  &  \coloneqq\sum_{k=1}^{K}\widehat{\beta}_{k}%
\psi_{k}(\cdot),\quad\widehat{v}_{j-1}(\cdot)\coloneqq\max(f_{j-1}%
(\cdot),\widehat{c}_{j-1}(\cdot)). \label{eq:TVR-classical-continuation-value}%
\end{align}
After this construction, we can either simply return the approximate value
$\widehat{v}_{0}(X_{0})$, or refine the estimate by simulating the expected
pay-off due to the nearly optimal stopping time,
\[
\widehat{\tau}=\min\left\{  j:0\leq j\leq\mathcal{J},\quad f_{j}%
(X_{j})>\widehat{c}_{j}(X_{j})\right\}  ,
\]
using newly generated independent samples from the process $X.$

The Longstaff--Schwartz algorithm is defined similarly, except that the
regression step~(\ref{eq:TVR-classical-regression}) does not use the
previously constructed value function $\widehat{v}_{j}$, but rather the nearly
optimal stopping time induced by $\widehat{c}_{j}, \ldots, \widehat
{c}_{\mathcal{J}}$. More precisely, the Longstaff--Schwartz algorithm goes as
follows: Initialize for $m$ $=$ $1,...,M,$ $\tau_{\mathcal{J}}^{(m)}%
\coloneqq\mathcal{J},$ $\widehat{c}_{\mathcal{J}}\coloneqq 0.$ If the
$\tau_{j}^{(m)}$ and $\widehat{c}_{j}$ are already constructed, iteratively
construct (backward in time)
\begin{align}
\widehat{\beta}^{(j-1)}  &  \coloneqq\argmin_{\beta\in\mathbb{R}^{K}}%
\sum_{m=1}^{M}\left(  f_{\tau_{j}^{(m)}}(X_{\tau_{j}^{(m)}}^{(m)})-\sum
_{k=1}^{K}\beta_{k}\psi_{k}(X_{j-1}^{(m)})\right)  ^{2},\label{SLS}\\
\widehat{c}_{j-1}(\cdot)  &  \coloneqq\sum_{k=1}^{K}\widehat{\beta}_{k}%
\psi_{k}(\cdot),\label{SLS1}\\
\text{If }f_{j-1}(X_{j-1}^{(m)})  &  >\widehat{c}_{j-1}(X_{j-1}^{(m)})\text{
\ \ then \ }\tau_{j-1}^{(m)}=j-1\text{ \ else \ \ }\tau_{j-1}^{(m)}=\tau
_{j}^{(m)}. \label{SLS2}%
\end{align}
In both algorithms, the regression step itself only relies on two random
variables, which we might as well denote by $(X,Y)\in\mathbb{R}^{d}%
\times\mathbb{R},$ living on some probability space $(\Omega,\mathcal{F}%
,\mathbb{P})$. Consider the problem of estimating the function $u:\mathbb{R}%
^{d}\rightarrow\mathbb{R},$ satisfying%
\begin{equation}
u(X)=\mathbb{E}\left[  Y|X\right]  . \label{form}%
\end{equation}

As indicated, we solve the least squares minimization problem%
\begin{equation}
\widehat{\beta}:=\underset{\beta\in\mathbb{R}^{K}}{\arg\inf}\sum_{m=1}%
^{M}\left(  Y^{(m)}-\sum_{k=1}^{K}\beta_{k}\psi_{k}\left(  X^{(m)}\right)
\right)  ^{2}, \label{lsq}%
\end{equation}
and consider the estimation%
\begin{equation}
\widehat{u}\left(  x\right)  =\sum_{k=1}^{K}\widehat{\beta}_{k}\psi_{k}\left(
x\right)  . \label{ste}%
\end{equation}
It is well-known that by defining the design matrix $\mathcal{N}\in
\mathbb{R}^{M\times K}$ by%
\[
\mathcal{N}_{mk}:=\psi_{k}\left(  X^{(m)}\right)  ,\text{ \ \ }%
m=1,...,M,\text{ }k=1,...,K,
\]
and the vector $\mathcal{Y}\in\mathbb{R}^{M}$ by%
\[
\mathcal{Y}_{m}=Y^{(m)},\text{ \ \ }m=1,...,M,
\]
that the solution to (\ref{lsq}) may be written as%
\begin{equation}
\widehat{\beta}=\frac{1}{M}\left(  \frac{1}{M}\mathcal{N}^{\top}%
\mathcal{N}\right)  ^{-1}\mathcal{N}^{\top}\mathcal{Y}, \label{bh}%
\end{equation}
provided that $\mathcal{N}$ has full rank $K.$ The latter is typically almost
surely the case when $K\leq M.$ Note that,%
\begin{align}
\frac{1}{M}\left[  \mathcal{N}^{\top}\mathcal{N}\right]  _{k,l=1,...,K}  &
\mathcal{=}\frac{1}{M}\sum_{m=1}^{M}\psi_{k}\left(  X^{(m)}\right)  \psi
_{l}\left(  X^{(m)}\right) \label{spe}\\
&  \approx\mathbb{E}\left[  \psi_{k}\left(  X\right)  \psi_{l}\left(
X\right)  \right]  .\nonumber
\end{align}
In general, the marginal distribution of $X$ is not explicitly known and the
inversion of the matrix $\frac{1}{M}\mathcal{N}^{\top}\mathcal{N}$ in
(\ref{bh}) is a main delicate issue since it has random nonnegative
eigenvalues that can be arbitrary close to zero by chance. Furthermore, the
computation of $\mathcal{N}$ requires about $KM$ function calls and the
computation of (\ref{bh}) requires about $K^{2}M$ elementary operations.

\section{The pseudo regression approach}

\label{Mark}

As an alternative to the well-known methods in \cite{LS2001} and \cite{TV2001}
we now propose a backward algorithm for approximating the continuation
functions $c_{j}$ (respectively $v_{j}$) by functions $\overline{c}_{j},$
(respectively $\overline{v}_{j},$) $j=\mathcal{J},...,0,$ in the present setup
that is based on \textit{pseudo regression}. Let us assume that we have chosen
a set of basis functions $\psi_{k}:\mathbb{R}^{d}\rightarrow\mathbb{R},$
$k=1,...,K,$ and a measure $\mu$ concentrated on $\mathcal{D}\subset
\mathbb{R}^{d},$ such that the Gram matrix $\mathcal{G}$ defined by
\begin{equation}
\mathcal{G}_{kl}:=\langle\psi_{k},\psi_{l}\rangle:=\int\psi_{k}(z)\psi
_{l}(z)\mu(dz) \label{gkl}%
\end{equation}
together with its inverse $\mathcal{G}^{-1}$ is explicitly known, or can be
efficiently computed. In the algorithm spelled out below we construct a set of
approximative continuation functions $\overline{c}_{j},$ $j=\mathcal{J}%
,...,0,$ which satisfy%
\[
\overline{c}_{j}(z)\approx c_{j}(z):=\mathbb{E}\left[  v_{j+1}(X_{j+1}%
^{j,z})\right]  =\mathbb{E}\left[  v_{j+1}(X_{j+1}^{0,\cdot})\,|\,X_{j}%
^{0,\cdot}=z\right]  .
\]
Moreover, it is assumed (for simplicity) that we are able to sample
trajectories $X_{\cdot}^{0,x}$ exactly.

The probability measure $\mu$ is used to measure the regression error,
cf.~\eqref{eq:L2-regression}, i.e., we try to minimize the difference between
$c_{j}$ and $\overline{c}_{j}$ in the sense of the $L^{2}(\mu)$-norm. From
that perspective, a natural choice of $\mu$ as induced by the problem at hand
would be the distribution of $X_{J+1}$, but that choice runs afoul of the
requirement that the Gram matrix is known explicitly. We shall see in
Section~\ref{sec:numer-exper} that the problem is not very sensitive to the
choice of $\mu$, such that we can often even choose a uniform (and simple)
reference measure $\mu$ for all $j$ without significant sacrifice in overall accuracy.

\subsection{Pseudo regression variant of Tsitsiklis--van Roy}

\label{alg}

We start with $\overline{v}_{\mathcal{J}}=v_{\mathcal{J}}=f_{\mathcal{J}}$ and
$\overline{c}_{\mathcal{J}}=c_{\mathcal{J}}=0.$ The backward iteration step
$j\rightarrow j-1$ works as follows: First generate $M$ i.i.d.~copies
$\mathcal{U}^{(m)},$ $m=1,\ldots,M$ with $\mathcal{U}^{(1)}\sim\mu.$
Simulate for $m=1,\ldots,M,$ the r.v. $X_{j}^{j-1,\mathcal{U}^{(m)}},$ and
consider the $M\times K$ matrix $\mathcal{M}^{(j)}$ defined by%
\[
\mathcal{M}_{mk}^{(j)}\coloneqq\psi_{k}\left(  \mathcal{U}^{(m)}\right)  .
\]
Define the vector $\mathcal{Y}^{(j)}\in\mathbb{R}^{M}$ by%
\begin{equation}
\mathcal{Y}_{m}^{(j)}\coloneqq\overline{v}_{j}\left(  X_{j}^{j-1,\mathcal{U}%
^{(m)}}\right)  . \label{caly}%
\end{equation}
Following~\eqref{eq:L2-regression}, the coefficients of the basis functions
are given by
\begin{equation}
\overline{\beta}^{(j)}\coloneqq\frac{1}{M}\mathcal{G}^{-1}\left(
\mathcal{M}^{(j)}\right)  ^{\top}\mathcal{Y}^{(j)} \label{pseudoreg}%
\end{equation}
and then we obtain the approximate continuation value and solution,
respectively, by
\begin{align}
&  \overline{c}_{j-1}(z)\coloneqq\sum_{k=1}^{K}\overline{\beta}_{k}^{(j)}%
\psi_{k}(z)\text{ \ \ and}\label{cont}\\
&  \overline{v}_{j-1}(z)\coloneqq\max\left(  f_{j-1}(z),\overline{c}%
_{j-1}(z)\right)  . \label{cont1}%
\end{align}
A pseudo-code representation of the algorithm is given in
Algorithm~\ref{alg:pseudo-bermudan}.

\begin{algorithm}
\DontPrintSemicolon
\KwData{$\mu, M, \psi_1, \ldots, \psi_K, \mathcal{G}, f_0, \ldots, f_{\mathcal{J}}$.}
\KwResult{Value function $\overline{v}_j$ and continuation value $\overline{c}_j$,
$j=0, \ldots, \mathcal{J}$.}
\Begin{
$\overline{v}_{\mathcal{J}} \longleftarrow v_{\mathcal{J}} =
f_{\mathcal{J}}$\;
$\overline{c}_{\mathcal{J}} \longleftarrow c_{\mathcal{J}} = 0$\;
\For{$m \longleftarrow 1$ \KwTo $M$}{
Generate $\mathcal{U}^{(m)} \sim \mu$\;
}
$\mathcal{M} \longleftarrow \left( \psi_k(\mathcal{U}^{(m)})
\right)_{\substack{m=1, \ldots, M\\ k=1, \ldots, K}} \in \mathbb{R}^{M\times K}$\;
\For{$j \longleftarrow \mathcal{J}$ \KwTo $1$}{
\For{$m \longleftarrow 1$ \KwTo $M$}{
Generate $X_j^{j-1,\mathcal{U}^{(m)}}$\;
\tcp{These r.v. are understood to be independent conditional $\mathcal{U}^{(m)}$}
}
$\mathcal{Y}^{(j)} \longleftarrow \left( \overline{v}_j\left(
X^{j-1,\mathcal{U}^{(m)}}_j \right) \right) _{m=1, \ldots, M} \in
\mathbb{R}^M$\;
$\overline{\beta}^{(j)} \longleftarrow \frac{1}{M} \mathcal{G}^{-1}
\mathcal{M}^\top \mathcal{Y}^{(j)}$\;
${\displaystyle \overline{c}_{j-1}(\cdot) \longleftarrow \sum_{k=1}^K
\overline{\beta}^{(j)}_k \psi_k(\cdot)}$\;
${\displaystyle \overline{v}_{j-1}(\cdot) \longleftarrow \max\left( f_{j-1}(\cdot), \,
\overline{c}_{j-1}(\cdot) \right)}$ \;
}
}
\caption{Pseudo regression variant of TV for Bermudan options}\label{alg:pseudo-bermudan}
\end{algorithm}

The pseudo regression algorithm for Bermudan options is related to the
well-known Tsitsiklis--van Roy algorithm (see \cite{TV2001}), but differs
essentially because of the pseudo regression step (\ref{pseudoreg}). In
contrast, Tsitsiklis--van Roy compute the coefficients (\ref{pseudoreg}) by
using standard global regression. Another striking difference is that in
Algorithm \ref{alg:pseudo-bermudan} the basis functions $\psi_{k}$ have to be
evaluated much less times, since only one sample of $M$ drawings from the
distribution $\mu$ serves for all exercise dates. The merits of standard
versus pseudo regression in a general setting are explained and discussed in
detail in Section~\ref{gvps} below.

\subsection{Pseudo regression variant of Longstaff--Schwartz}

\label{sec:pseudo-regr-LS}

In order to obtain a pseudo regression variant of the Longstaff--Schwartz
algorithm we modify the backward construction of the approximative
continuation functions $\overline{c}_{j},$ $j=\mathcal{J},...,0$ (initialized
with $\overline{c}_{\mathcal{J}}=0$ again) in the following way. Let us assume
that $\overline{c}_{j},...,\overline{c}_{\mathcal{J}}$ are constructed.
Simulate for $m=1,...,M$ at time $j-1$ the trajectory
\begin{equation}
X_{r}^{j-1,\mathcal{U}^{(m)}},\text{ \ \ }r=j,...,\mathcal{J}, \label{sh}%
\end{equation}
and modify (\ref{caly}) to%
\begin{align}
\mathcal{Y}_{m}^{(j)}\coloneqq  &  f_{\tau}(X_{\tau}^{j-1,\mathcal{U}^{(m)}%
}),\text{ \ \ where}\label{LSc}\\
\tau &  \equiv\min\left\{  r:r\geq j,\text{ }f_{r}(X_{r}^{j-1,\mathcal{U}%
^{(m)}})\geq\overline{c}_{r}\left(  X_{r}^{j-1,\mathcal{U}^{(m)}}\right)
\right\}  .\nonumber
\end{align}
Then compute (\ref{pseudoreg}) and set $\overline{c}_{j-1}(z)$ according to
(\ref{cont}). The corresponding modification of Algorithm
\ref{alg:pseudo-bermudan} is obvious.\smallskip

At the first glance this procedure is significantly more costly. However, if
the chain $X$ is autonomous, which we may assume w.l.o.g. in fact, we simulate
first%
\[
X_{r}^{0,\mathcal{U}^{(m)}},\text{ \ \ }r=0,...,\mathcal{J},
\]
and then take in (\ref{sh})%
\begin{align}
\label{shifttime}X_{r}^{j-1,\mathcal{U}^{(m)}}=X_{r-j+1}^{0,\mathcal{U}^{(m)}%
},\text{ \ \ }r=j,...,\mathcal{J}.
\end{align}
So, for the autonomous case, one set of full trajectories, just as in the
standard LS algorithm, is sufficient for this algorithm as well.

\section{Accuracy analysis of pseudo regression}

\label{gvps}

In the next section we analyze an alternative and potentially more efficient
pseudo regression procedure for computing $\mathbb{E}\left[  Y|X\right]  ,$
i.e. (\ref{form}), given that we may sample $Y$ from its conditional
distribution given $X$ (although we generally do not know $\mathbb{E}\left[
Y|X\right]  $ explicitly of course).

\subsection{A general framework}

\ Suppose that in (\ref{form}) it is possible to sample $Y$ from its
conditional distribution given $X,$ say $\nu\left(  dy|X\right)  .$ A
canonical example is the setup in Section \ref{Mark} where%
\[
X=X_{j}^{0,x}\text{ \ \ and \ \ }Y=g\left(  X_{j+1}^{0,x}\right)  =g\left(
X_{j+1}^{j,X_{j}^{0,x}}\right)  ,
\]
for some arbitrary $x.$ Let us consider a random variable $\mathcal{U}$ with
values in some domain $\mathcal{D}\subset\mathbb{R}^{d},$ distributed
according to some probability measure $\mu(dz)$ concentrated on $\mathcal{D}.$
We then generate i.i.d. copies $\mathcal{U}^{(m)},$ $m=1,...,M$ of
$\mathcal{U},$ and sample for each $m=1,...,M,$ independently $Y^{(m)}$ from
$\nu\left(  dy|\mathcal{U}^{(m)}\right)  .$ Then define the vector
$\mathcal{Y\in}\mathbb{R}^{M}$ as%
\[
\mathcal{Y}:=\left[  Y^{(1)},...,Y^{(M)}\right]  ^{\top}.
\]
Now for a linearly independent system $\left(  \psi_{k}:k=1,2,...\right)  ,$
with%
\[
\int\psi_{k}^{2}(z)\mu(dz)<\infty,
\]
consider the $M\times K$ matrix%
\[
\mathcal{M}_{mk}:=\psi_{k}\left(  \mathcal{U}^{(m)}\right)  .
\]
Assuming that we know explicitly the matrix $\mathcal{G}$ defined by the
scalar products $\mathcal{G}_{kl}:=\langle\psi_{k},\psi_{l}\rangle$ (cf.
(\ref{gkl})), we now compute the pseudo regression coefficients
\begin{equation}
\overline{\beta}=\frac{1}{M}\mathcal{G}^{-1}\mathcal{M}^{\top}\mathcal{Y},
\label{bt}%
\end{equation}
and consider the pseudo regression approximation%
\begin{equation}
\overline{u}(z)=\sum_{k=1}^{K}\overline{\beta}_{k}\psi_{k}\left(  z\right)
\approx\mathbb{E}\left[  Y\,|\,\mathcal{U}=z\right]  ,\text{ \ \ }%
z\in\mathcal{D}. \label{ut}%
\end{equation}

\bigskip Clearly, the difference with standard regression is that the random
matrix $\frac{1}{M}\mathcal{N}^{\top}\mathcal{N}$ in (\ref{bh}) is replaced by
$\mathcal{G}$ in view of (\ref{spe}). In general $\mathcal{G}^{-1}$ can be
pre-computed outside the Monte Carlo simulation with arbitrary accuracy or is
explicitly known due to a suitable choice of the system $\left(  \psi
_{k}:k=1,2,...\right)  $ and the measure $\mu.$ So the computation of
(\ref{bt}) only involves $KM$ elementary operations and no random matrix
inversion is required. Moreover, naturally, we may assume w.l.o.g. that the
system $\left(  \psi_{k}:k=1,2,...\right)  $ is an orthonormal system with
respect to $L_{2}\left(  \mathcal{D},\mu\right)  $ and then\ (\ref{bt})
simplifies to%
\[
\overline{\beta}=\frac{1}{M}\mathcal{N}^{\top}\mathcal{Y}.
\]

\subsection{Accuracy analysis of the regression}

For the convergence properties of the pseudo-regression method we could
basically refer to \cite{Ank2017,BBRRS18}, where pseudo regression is applied
in the context of global solutions for random PDEs. For the convenience of the
reader, however, let us here recap the analysis in condensed form, consistent
with the present terminology and a somewhat less involved setup.

\begin{theorem}
\label{psth} (Accuracy pseudo regression) Suppose that in (\ref{form})%
\begin{gather*}
\left\vert u(z)\right\vert \leq D\text{ \ \ and \ \ }\operatorname{Var}\left[
Y\,|\,X=z\right]  <\sigma^{2},\text{ \ \ for all }z\in\mathcal{D},\\
0<\underline{\lambda_{\min}}\leq\lambda_{\min}\left(  \mathcal{G}^{K}\right)
\leq\lambda_{\max}\left(  \mathcal{G}^{K}\right)  \leq\overline{\lambda_{\max
}},\text{ \ \ for all }K=1,2,...,
\end{gather*}
where $\lambda_{\min}\left(  \mathcal{G}^{K}\right)  ,$ and $\lambda_{\max
}\left(  \mathcal{G}^{K}\right)  ,$ denote the smallest, respectively largest,
eigenvalue of the positive symmetric matrix $\mathcal{G}.$ Then it holds,%
\begin{align}
&  \mathbb{E}\int_{\mathcal{D}}\left\vert \overline{u}(z)-u(z)\right\vert
^{2}\mu(dz)\label{th1}\\
&  \leq\frac{\overline{\lambda_{\max}}}{\underline{\lambda_{\min}}}\left(
\sigma^{2}+D^{2}\right)  \frac{K}{M}+\underset{w\,\in\,\spn \{\psi
_{1},...,\psi_{K}\}}{\inf}\int_{\mathcal{D}}\left\vert w(z)-u(z)\right\vert
^{2}\mu(dz).\nonumber
\end{align}

\end{theorem}

\noindent The proof of Theorem~\ref{psth} is provided in
Appendix~\ref{sec:proof-theorem-psth}.

It is interesting to compare Theorem~\ref{psth} with a corresponding theorem
that holds for the standard regression estimate (\ref{ste}):

\begin{theorem}
\label{rgth} (Accuracy standard regression) Suppose that,%
\[
\left\vert u(x)\right\vert \leq D\text{ \ \ and \ \ }\operatorname{Var}\left[
Y\,|\,X=x\right]  <\sigma^{2},\text{ \ \ for all }x\in\mathbb{R}^{d},
\]
then for%
\[
\widetilde{u}_{D}(x)=\left\{
\begin{array}
[c]{c}%
\widetilde{u}(x)\text{ \ \ if \ \ }\left\vert \widetilde{u}(x)\right\vert \leq
D\\
D\text{ \ \ if \ \ }\widetilde{u}(x)>D\\
-D\text{ \ \ if \ \ }\widetilde{u}(x)<-D
\end{array}
\right.
\]
and some universal constant $c>0,$ it holds that%
\begin{align}
&  \mathbb{E}\int\left\vert \widetilde{u}_{D}(x)-u(x)\right\vert ^{2}\mu
_{X}(dx)\label{Gyth}\\
&  \leq c\max\left(  \sigma^{2},D^{2}\right)  \frac{\left(  1+\ln M\right)
K}{M}+8\underset{w\,\in\,\spn\{\psi_{1},...,\psi_{K}\}\,}{\inf}\int
_{\mathcal{D}}\left\vert w(x)-u(x)\right\vert ^{2}\mu_{X}(dx),\nonumber
\end{align}
where $\mu_{X}$ denotes the distribution of $X$ in (\ref{form}).
\end{theorem}

The proof of Theorem~\ref{rgth} is much more complicated than the proof of
Theorem~\ref{psth} and relies heavily on uniform laws of large numbers from
the theory of empirical processes. For details see \cite{Gy2002}.

\subsection{Convergence of the pseudo LS and pseudo TV algorithm}

In this section we investigate the convergence properties of the pseudo
Longstaff--Schwartz and pseudo Tsitsiklis--van Roy algorithm. Let us first
consider the pseudo LS method which is the more involved one in fact. We
follow similar lines as in \cite{Z} and in \cite{BelSchBMV} on optimal
stopping in the context of interacting particle systems. More specifically, we
consider the algorithm based on (\ref{sh}), where for every exercise date the
sample (\ref{sh}) is simulated independently, and consider the information set%
\[
\mathcal{G}_{j}:=\sigma\left\{  \mathbf{X}^{j;M},\ldots,\mathbf{X}%
^{\mathcal{J}-1;M}\right\}  \text{ with }\mathbf{X}^{j;M}:=\left(
X_{r}^{j,\mathcal{U}^{(m)},m},\text{ \ \ }r=j,...,\mathcal{J},\text{
}m=1,...,M\right)  .
\]
Let us define for a generic dummy trajectory $\left(  X_{l}\right)
_{l=0,\ldots,\mathcal{J}}$ corresponding to the (exact) solution independent
of $\mathcal{G}_{j},$
\begin{equation}
\widetilde{c}_{j}(x):=\mathbb{E}_{\mathcal{G}_{j+1}}\left[  \left.
f_{\overline{\tau}_{j+1}}\left(  X_{\overline{\tau}_{j+1}}\right)  \right\vert
X_{j}=x\right]  , \label{ctil}%
\end{equation}
where $\overline{\tau}_{\mathcal{J}}=\mathcal{J},$ and%
\[
\overline{\tau}_{j}:=j\,1_{\bigl\{f_{j}(X_{j})\geq\overline{c}_{j}%
(X_{j})\bigr\}}+\overline{\tau}_{j+1}1_{\bigl\{f_{j}(X_{j})<\overline{c}%
_{j}(X_{j})\bigr\}}.
\]
It is important to note that, in (\ref{ctil}), $\widetilde{c}_{j}\left(
\cdot\right)  $ is a $\mathcal{G}_{j+1}$-measurable random function while the
estimation $\overline{c}_{j}\left(  \cdot\right)  $ is a $\mathcal{G}_{j}%
$-measurable one as the construction of $\overline{c}_{j}$ also depends on
$\mathbf{X}^{j;M},$ see (\ref{cont}) connected with (\ref{LSc}). After
proceeding backwardly from $j=\mathcal{J}$ down to $j=1,$ we thus have a
sequence of approximative continuation functions $\overline{c}_{j}\left(
\cdot\right)  ,$ and a sequence of corresponding conditional expectations
$\widetilde{c}_{j}\left(  \cdot\right)  .$ The convergence analysis for the
pseudo LS method is based on the following lemma (cf. Lemma~5 in
\cite{BelSchBMV}).

\begin{lemma}
\label{lem23} For the conditional expectations (\ref{ctil}) we have that,%
\begin{equation}
\left\Vert \widetilde{c}_{j}-c_{j}\right\Vert _{L_{p}(\mu)}\leq\sum
_{l=j+1}^{\mathcal{J}-1}\left\Vert \overline{c}_{l}-c_{l}\right\Vert
_{L_{p}(\mu_{j,l})} \label{eq:bound_1}%
\end{equation}
with $p\geq1,$ $\mu_{j,l}$ being the distribution of $X_{l}^{j,\mathcal{U}},$
$1\leq j\leq l\leq\mathcal{J},$ $\mathcal{U}\sim\mu,$ and $c_{j},$ being the
true continuation functions.
\end{lemma}

The proof is almost identical with the proof of the similar Lemma~5 in
\cite{BelSchBMV}. For the convenience of the reader, it is given in
Appendix~\ref{sec:proof-lemma-reflem23}.

\begin{remark}
Note that the inequality (\ref{eq:bound_1}) involves $\mathcal{G}_{j+1}%
$-measurable objects. It is also interesting to compare (\ref{eq:bound_1})
with similar (though different) inequalities in \cite{Z}.
\end{remark}

We now state our convergence theorem connected with the pseudo
Longstaff--Schwartz algorithm. The proof is given in
Appendix~\ref{sec:proof-theorem-refeq}.

\begin{theorem}
\label{thm_main} Let us assume that the conditions of Theorem \ref{psth} are
fulfilled. In particular, we assume that the cash-flows $f_{j}$ are uniformly
bounded, i.e. $0\leq f_{j}\leq D$ for $j=0,\ldots,\mathcal{J},$ and some
$D>0.$ Since then, as a consequence, also $0\leq c_{j}\leq D,$ we may moreover
assume that $0\leq\overline{c}_{j}\leq D,$ for $j=0,\ldots,\mathcal{J}.$ Let
us further assume, somewhat more general, sampling densities $\mu_{j}$ that
may depend on $j$ in the Longstaff-Schwartz version of Algorithm
\ref{alg:pseudo-bermudan}, which moreover satisfy $\mu_{j}>0$ for
$j=0,\ldots,\mathcal{J}-1,$ and%
\[
\mathcal{R}_{\infty}:=\max_{0\leq j<l<\mathcal{J}}\sup_{x\in\mathbb{R}^{d}%
}\frac{\mu_{j,l}(x)}{\mu_{l}(x)}<\infty,
\]
where $X_{l}^{j,\mathcal{U}_{j}}\sim\mu_{j,l}.$ For a generic measure $\nu,$
the norm%
\[
\left\Vert \cdot\right\Vert _{L_{2}(\nu\otimes\mathbb{P})}^{2}:=\mathbb{E}%
\left[  \left\Vert \cdot\right\Vert _{L_{2}(\nu)}^{2}\right]  ,
\]
is defined due to the unconditional expectation with respect to the
\textquotedblleft all in\textquotedblright\ probability measure $\mathbb{P}.$
One then has for natural numbers $K,M,$ and $j=0,\ldots,\mathcal{J}-1,$%
\begin{align}
&  \left\Vert \overline{c}_{j}-c_{j}\right\Vert _{L_{2}(\mu_{j}\otimes
\mathbb{P})}\nonumber\\
&  \leq\eta\varepsilon_{j,M,K}\left(  1+\mathcal{R}_{\infty}^{1/2}\left(
\eta+1\right)  \right)  ^{\mathcal{J}-j-1},\nonumber\\
\text{where}\qquad\varepsilon_{j,M,K}  &  :=\sqrt{\frac{K}{M}}+\max_{j\leq
l<\mathcal{J}}\underset{w\,\in\,\text{\textsf{span}}\{\psi_{1},...,\psi_{K}%
\}}{\inf}\left\Vert c_{l}-w\right\Vert _{L_{2}(\mu_{j})}, \label{eq: main}%
\end{align}
$\mathcal{U}_{j}\sim\mu_{j},$ $0\leq j<\mathcal{J},$ and $\eta>0$ is some
constant not depending on $K,M,R,$ and the choice of the densities $\mu_{j}.$
\end{theorem}

Let us now consider the convergence of the pseudo TV method. For a generic
exact (dummy) solution $\left(  X_{l}\right)  _{l=0,\ldots,\mathcal{J}}$
independent of $\mathcal{G}_{j},$ we now re-define (\ref{ctil}) as
\begin{equation}
\widetilde{c}_{j}(x):=\mathbb{E}_{\mathcal{G}_{j+1}}\left[  \left.
\text{\ }\overline{v}_{j+1}\left(  X_{j+1}\right)  \right\vert X_{j}=x\right]
, \label{tiltv}%
\end{equation}
where again, in (\ref{tiltv}), $\widetilde{c}_{j}\left(  \cdot\right)  $ is a
$\mathcal{G}_{j+1}$-measurable random function while the estimation
$\overline{c}_{j}\left(  \cdot\right)  $ is a $\mathcal{G}_{j}$-measurable one
that is now constructed via (\ref{cont}) and (\ref{cont1}). The convergence of
the pseudo TV method is based on the next lemma.

\begin{lemma}
\label{lem23TV} For the conditional expectations (\ref{ctil}) we have that,%
\begin{equation}
\left\Vert \widetilde{c}_{j}-c_{j}\right\Vert _{L_{p}(\mu)}\leq\left\Vert
\overline{c}_{j+1}-c_{j+1}\right\Vert _{L_{p}(\mu_{j,j+1})} \label{boundTV}%
\end{equation}
with $p\geq1,$ $\mu_{j,j+1}$ being the distribution of $X_{j+1}^{j,\mathcal{U}%
},$ $1\leq j<\mathcal{J},$ $\mathcal{U}\sim\mu,$ and $c_{j},$ being the true
continuation functions.
\end{lemma}

The proof is somewhat simpler than the proof of Lemma~\ref{lem23} and given in
Appendix~\ref{sec:proof-lemma-reflem23TV}. For the pseudo Tsitsiklis--van Roy
algorithm we now have the following convergence theorem, proved in
Appendix~\ref{sec:proof-theorem-refeqTV}.

\begin{theorem}
\label{thm_mainTV} Let us consider the same assumptions and notation as in
Theorem~\ref{thm_main}, but, with now $\mu_{j},$ $j=1,...,\mathcal{J},$ being
the sampling densities (generally depending on $j$) in the Tsitsiklis--van Roy
version of Algorithm \ref{alg:pseudo-bermudan}. If all $\mu_{j}>0$ and
satisfy,%
\begin{equation}
\mathcal{R}_{+}:=\max_{0\leq j<\mathcal{J}-1}\sup_{x\in\mathbb{R}^{d}}%
\frac{\mu_{j,j+1}(x)}{\mu_{j+1}(x)}<\infty, \label{max}%
\end{equation}
then for natural numbers $K,M,$ and $j=0,\ldots,\mathcal{J}-1,$%
\begin{align}
&  \left\Vert \overline{c}_{j}-c_{j}\right\Vert _{L_{2}(\mu_{j}\otimes
\mathbb{P})}\leq\eta\varepsilon_{j,M,K}\frac{\left(  \mathcal{R}_{+}%
^{1/2}\left(  \eta+1\right)  \right)  ^{\mathcal{J}-j}-1}{\mathcal{R}%
_{+}^{1/2}\left(  \eta+1\right)  -1},\label{eq: mainTV}\\
\text{where}\qquad\varepsilon_{j,M,K}  &  :=\sqrt{\frac{K}{M}}+\max_{j\leq
l<\mathcal{J}}\underset{w\,\in\,\text{\textsf{span}}\{\psi_{1},...,\psi_{K}%
\}}{\inf}\left\Vert c_{l}-w\right\Vert _{L_{2}(\mu_{j})},\nonumber
\end{align}
$\mathcal{U}_{j}\sim\mu_{j},$ $0\leq j<\mathcal{J},$ and $\eta>0$ is some
constant not depending on $K,M,R,$ and the choice of the densities $\mu_{j}.$
\end{theorem}

\subsection{On the choice of the measures $\mu_{l}$}

In the formulation of the above results it is assumed that the state space of
the underlying process $X$ is $\mathbb{R}^{d}$ but, naturally, the results
apply also if $X$ runs through some open subset of $\mathbb{R}^{d},$
$\mathbb{R}_{>0}^{d}$ for example.

Let us take $\mu_{l}\sim X_{l}^{0,\mathcal{U}_{0}}.$ Then we have $\mu_{j,l}$
$\sim$ $X_{l}^{j,\mathcal{U}_{j}}$ $\sim$ $X_{l}^{0,\mathcal{U}_{0}}$ $\sim$
$\mu_{l},$ and thus $\mathcal{R}_{\infty}=\mathcal{R}_{+}=1$ in
Theorem~\ref{thm_main} and Theorem~\ref{thm_mainTV}, respectively. For the
accuracy estimates we then obtain,
\begin{equation}
\left\Vert \overline{c}_{0}-c_{0}\right\Vert _{L_{2}(\mu_{0}\otimes
\mathbb{P})}\leq\eta\varepsilon_{0,M,K}\left(  \eta+2\right)  ^{\mathcal{J}%
-1}, \label{lb}%
\end{equation}
and%
\begin{equation}
\left\Vert \overline{c}_{0}-c_{0}\right\Vert _{L_{2}(\mu_{0}\otimes
\mathbb{P})}\leq\varepsilon_{0,M,K}\left(  \left(  \eta+1\right)
^{\mathcal{J}}-1\right)  , \label{lbTV}%
\end{equation}
respectively. This hypothetical choice of sampling measures shows that in
principle the accuracy bounds (\ref{lb}) and (\ref{lbTV}) are attainable for
LS and TV respectively. Actually, this touches the cardinal point in
Glasserman-Yu (2002) where, loosely speaking, one of the advices was to search
at each step for basis functions that are orthogonal with respect to the
distribution of the underlying process. However, in practice this is rarely
possible, apart from the case of an underlying (multidimensional)
Black-Scholes model. \newline

In fact, this paper proposes to go beyond Glasserman-Yu (2002): Based on
Theorems \ref{thm_main} and \ref{thm_mainTV}, we suggest to search for
\textquotedblleft suitable\textquotedblright\ densities $\mu_{0}%
,...,\mu_{\mathcal{J}-1}$ such that, on the one hand, the densities $\mu_{l}$
are in some sense close to the densities of $X_{l}^{0,\mathcal{U}_{0}},$ for
$l=0,...,\mathcal{J},$ such that $\mathcal{R}_{\infty}<\infty,$ and on the
other hand, they need to be such that for each exercise date $l,$ a
$K$-dimensional system of basis functions $\Psi_{l}$ is available such that
$\langle\psi,\psi^{\prime}\rangle$ is known for all $\psi,\psi^{\prime}\in
\Psi_{l},$ or even better $\langle\psi,\psi^{\prime}\rangle=\delta_{\psi
\psi^{\prime}}.$

Let us choose the $\mu_{l}$ such that%
\[
\sup_{x\in\mathbb{R}^{d}}\frac{\mu_{l-1,l}(x)}{\mu_{l}(x)}=\sup_{x\in
\mathbb{R}^{d}}\frac{\int\mu_{l-1}(x_{l})p_{l-1,l}(x_{l-1},x)dx_{l-1}}{\mu
_{l}(x)}\leq\mathcal{R}_{l},\text{ \ \ }l=0,...,\mathcal{J}-1,
\]
where we denote by $p_{l,l^{\prime}}(x_{l},x_{l^{\prime}})$ the density of
$X_{l^{\prime}}^{0,x_{l}}$ for $0$ $\leq$ $l<$ $l^{\prime}$ $<$ $\mathcal{J}.$
We then have that,%
\begin{align*}
\mu_{j,l}(x)  &  =\int dx_{j}\mu_{j}(x_{j})p_{j,l}(x_{j},x_{l})\\
&  =\int dx_{j+1}p_{j+1,l}(x_{j+1},x_{l})\int dx_{j}\mu_{j}(x_{j}%
)p_{j,j+1}(x_{j},x_{j+1})\\
&  \leq\mathcal{R}_{j+1}\int dx_{j+1}\mu_{j+1}(x_{j+1})p_{j+1,l}(x_{j+1}%
,x_{l})\\
&  \leq...\leq\mathcal{R}_{j+1}\mathcal{R}_{j+2}\cdot\cdot\cdot\mathcal{R}%
_{l}\mu_{l}(x),
\end{align*}
and so we may take%
\[
\mathcal{R}_{\infty}:=\max_{0\leq j\leq l<\mathcal{J}}\mathcal{R}%
_{j+1}\mathcal{R}_{j+2}\cdot\cdot\cdot\mathcal{R}_{l}=\mathcal{R}%
_{0}\mathcal{R}_{1}\cdot\cdot\cdot\mathcal{R}_{\mathcal{J}-1}<\infty,
\]
and $\mathcal{R}_{+}<\infty$ in Theorem~\ref{thm_main} and
Theorem~\ref{thm_mainTV}, respectively.

\begin{example}
Let $X$ be given by an It\^{o} diffusion with state space $\mathbb{R}_{+}%
^{d},$ and let $q_{l,l^{\prime}}(y_{l},y_{l^{\prime}})$ be the density of the
(log-price) process $L_{l},$ $l=0,...,\mathcal{J},$ defined by%
\[
L_{l}:=\ln\left[  X_{l}\right]  \text{ \ \ with \ \ }\ln\left[  x\right]
:=\left[  \ln(x^{i})\right]  _{i=1,...,d}\text{ \ for \ }x\in\mathbb{R}%
_{+}^{d}.
\]
Suppose that $q_{l-1,l}(y_{l},y_{l^{\prime}})$ is sub-Gaussian with some
(possibly complicated), correlation structure. A typical situation would be
the case where
\begin{multline}
q_{l-1,l}(y_{l-1},y_{l})\leq\\
\mathcal{R}_{l}\frac{1}{\sqrt{\left(  2\pi\right)  ^{d}\det(\Sigma^{(l)})}%
}\exp\left[  -\frac{1}{2}(y_{l-1}-y_{l})^{T}\left(  \Sigma^{(l)}\right)
^{-1}(y_{l-1}-y_{l})\right]  \\
=:\mathcal{R}_{l}\widehat{q}_{l-1,l}(y_{l-1},y_{l}),\label{qr}%
\end{multline}
for a simple covariance matrix $\Sigma^{(l)},$ for instance a diagonal matrix.
Let $\mu_{j}$ be the density of the sampling random variable $\mathcal{U}_{j}$
$\in\mathbb{R}_{+}^{d},$ given by%
\[
\mathcal{U}_{j}:=\exp\left[  \mathcal{L}_{j}\right]  ,\text{ \ \ with \ }%
\exp\left[  y\right]  :=\left[  \exp(y^{i})\right]  _{i=1,...,d}\text{ \ for
\ }y\in\mathbb{R}^{d},
\]
where $\mathcal{L}_{j}$ is sampled from a density $\upsilon_{j}.$ For an
arbitrary nonnegative Borel function $f$ on $\mathbb{R}_{+}^{d}$ one has%
\[
\int f(x)\mu_{j}(x)dx=\int f(\exp[y])\upsilon_{j}(y)dy=\int f(x)\upsilon
_{j}(\ln\left[  x\right]  )%
{\displaystyle\prod\limits_{k=1}^{d}}
\frac{1}{x^{k}}dx,
\]
whence%
\begin{equation}
\mu_{j}(x)=\upsilon_{j}(\ln\left[  x\right]  )%
{\displaystyle\prod\limits_{k=1}^{d}}
\frac{1}{x^{k}},\label{mutr}%
\end{equation}
and similarly one has%
\begin{equation}
p_{l-1,l}(x_{l-1},x)=q_{l-1,l}(\ln\left[  x_{l-1}\right]  ,\ln\left[
x\right]  )%
{\displaystyle\prod\limits_{k=1}^{d}}
\frac{1}{x^{k}}\label{ptr}%
\end{equation}
for the one step transition density of $X.$ Now let us take $\upsilon
_{1}(y_{1}):=\widehat{q}_{0,1}(y_{0},y_{1}),$ and recursively,
\[
\upsilon_{l}(y_{l})=\int dy_{l-1}\upsilon_{l-1}(y_{l-1})\widehat{q}%
_{l-1,l}(y_{l-1},y_{l}).
\]
We then have by (\ref{qr}),(\ref{mutr}), and (\ref{ptr})
\begin{align*}
&  \int\mu_{l-1}(x_{l-1})p_{l-1,l}(x_{l-1},x)dx_{l-1}=\int\upsilon_{l-1}%
(\ln\left[  x_{l-1}\right]  )%
{\displaystyle\prod\limits_{k=1}^{d}}
\frac{1}{x_{l-1}^{k}}dx_{l-1}\cdot\\
&  \cdot q_{l-1,l}(\ln\left[  x_{l-1}\right]  ,\ln\left[  x\right]  )%
{\displaystyle\prod\limits_{r=1}^{d}}
\frac{1}{x^{r}}\\
&  =\int\upsilon_{l-1}(y_{l-1})dy_{l-1}q_{l-1,l}(y_{l-1},\ln\left[  x\right]
)%
{\displaystyle\prod\limits_{r=1}^{d}}
\frac{1}{x^{r}}\leq\mathcal{R}_{l}\upsilon_{l}(\ln\left[  x\right]  )%
{\displaystyle\prod\limits_{r=1}^{d}}
\frac{1}{x^{r}}=\mathcal{R}_{l}\mu_{l}(x).
\end{align*}

\end{example}

\section{Computational cost}

\label{sec:computational-cost}

We will now discuss the advantages and disadvantages of the two different
approaches for various use-cases, both in the context of the Tsitsiklis--van
Roy algorithm and the Longstaff--Schwartz algorithm. The main issue is, of
course, the relation between computational work and accuracy. Comparing
Theorems~\ref{psth} and~\ref{rgth}, we see that the error as function of the
number of basis functions $K$, the choice of basis functions $\psi_{1},
\ldots, \psi_{K}$ and the number of samples $M$ is roughly equivalent for both methods.

\begin{remark}
\label{rem:1} We ignore the different constants as well as the additional $\ln
M$ term in Theorem~\ref{rgth}. In practice, different constants may, of
course, have drastic effects ion run-time, which is why the numerical
experiments presented in Section~\ref{sec:numer-exper} are crucial. {A more
subtle difference is related to the choice of the measure with respect to
which the error is calculated.} Also we note that we only focus on the cost of
computing the functions $c_{j},$ respectively $v_{j},$ as the other aspects of
the computation have negligible cost, independent of the chosen regression method.
\end{remark}

Let us recall our general setting: we are given a cash-flow process $Z_{j} =
f_{j}(X_{j})$, $j = 0, \ldots, \mathcal{J}$, which is based on an
$\mathbb{R}^{d}$-valued Markov process $X_{j}$, $j=0, \ldots, \mathcal{J}$,
and we would like to compute the corresponding Bermudan option price. In the
following, we need to make certain assumptions on the simulation.

\begin{assumption}
\label{ass:exact-simulation} We can exactly simulate the Markov process $X$.
More precisely, given a sample of $X_{j}$, we can simulate a sample of
$X_{j+1}$, $j=0, \ldots, \mathcal{J}-1$, exactly at cost normalized to one.
\end{assumption}

Assumption~\ref{ass:exact-simulation} seems to restrict us to simple models
such as Black-Scholes or Bachelier, for which exact simulation is easily
possible, but note that any discretization error would be expected to effect
both regression algorithms in the same way, both with respect to accuracy and
with respect to computational cost. Therefore, we think that
Assumption~\ref{ass:exact-simulation} is justified.

\begin{remark}
[Cost model]\label{rem:cost-model} In the discussions of computational cost,
all estimates shall be understood as counting the number of \emph{function
evaluations}. More precisely, each of the following operations incurs one unit cost:

\begin{itemize}
\item Generating one sample of $X_{j+1}$ conditional on $X_{j}$;

\item Evaluating a basis function $\psi_{k}$ at one point $x$;

\item An elementary floating point calculation such as a product between two
floating point numbers.
\end{itemize}

Of course, these operations incur very different computational costs in
practice. However, note that it is very difficult to realistically bound true
computational times any way. These may heavily depend on hardware features
(e.g., cache misses), and, in particular, on the implementation details.
\end{remark}

\subsection{Tsitsiklis--van Roy algorithm}

\label{sec:tsitsiklis-van-roy}

With Assumption~\ref{ass:exact-simulation}, we can already describe the
computational work of the standard regression algorithm.

\begin{proposition}
[Computational cost of standard regression]%
\label{prop:cost-standard-regression} The computational cost of the standard
regression satisfies
\[
\mathcal{C}_{\text{reg}} = \mathcal{O}\left(  \mathcal{J}( M K^{2} + K^{3})
\right)  .
\]

\end{proposition}

\begin{proof}
This result is, of course, very well known. The dominating terms for the
computational cost the computation of the random matrix $\mathcal{N}^{\top
}\mathcal{N}$ and the computation of the coefficient $\widetilde{\beta}$ by,
e.g., LU or Cholesky decomposition. Both operations have to be recomputed for
each exercise time $j = 0, \ldots, \mathcal{J}-1$.
\end{proof}

For the pseudo regression approach we will operate under

\begin{assumption}
\label{ass:explicit-G} The basis functions $\psi_{1}, \ldots, \psi_{K}$ are
chosen such that the matrix $\mathcal{G}$ is given explicitly.
\end{assumption}

The assumption is most easily satisfied by choosing the basis function to be
orthonormal polynomials w.r.t.~$\mu$. Then we obtain

\begin{proposition}
[Computational cost of pseudo regression]\label{prop:cost-pseudo-inhom} The
computational cost of pseudo regression under
Assumptions~\ref{ass:exact-simulation} and~\ref{ass:explicit-G} satisfies
\[
\mathcal{C}_{\text{pseudo}} = \mathcal{O}\left(  \mathcal{J} M K + \mathcal{J}
K^{2} + K^{3} \right)  .
\]
If, in addition, the basis functions are orthonormal w.r.t.~$\mu$, then the
cost instead satisfies
\[
\mathcal{C}_{\text{pseudo}} = \mathcal{O}\left(  \mathcal{J} M K \right)  .
\]

\end{proposition}

\begin{proof}
First we need to compute the LU decomposition of the matrix $\mathcal{G}$, at
cost proportional to $K^{3}$---independent of $\mathcal{J}$. {We also need to
simulate the random variables $\mathcal{U}$ and set up the matrix
$\mathcal{M}$ at cost $\mathcal{O}(MK)$.} In each iteration of the algorithm,
we then need to {simulate the vector $\mathcal{Y}$ and} multiply
$\mathcal{M}^{\top}\mathcal{Y}$ at cost proportional to $MK$. Finally,
assembling the solution of the linear system $\mathcal{G} \overline{\beta} =
\frac{1}{M} \mathcal{M}^{\top}\mathcal{Y}$ incurs costs proportional to
$K^{2}$.

In the orthonormal case, we have $\mathcal{G} = \mathcal{G}^{-1} =
\operatorname{Id}_{K}$, and the cost of setting up $\mathcal{M}$ and
multiplying $\mathcal{M}^{\top}\mathcal{Y}$ becomes dominant.
\end{proof}

In practice, even better cost savings are possible under

\begin{assumption}
\label{ass:homog-markov} The Markov process $X$ is homogeneous in time, i.e.,
the conditional distribution of $X_{j+1}$ given $X_{j}$ does not depend on $j$.
\end{assumption}

This condition is very often satisfied in financial models, and it has drastic
implications for the pseudo regression algorithm (but not for the standard
regression). Indeed, since the conditional distribution does not depend on
$j$, and we always re-sample the starting points (at step $j$) from the same
distribution $\mu$---instead of the distribution of $X_{j}$---, we can simply
use the \emph{same} samples for setting up $\mathcal{Y}$ for each time-step
in~\eqref{bt}. Formally, the asymptotic cost does not change compared to
Proposition~\ref{prop:cost-pseudo-inhom}, but in practice we do observe major
effects due to decreasing constants.

\begin{remark}
\label{rem:non-exact-G} It is well-understood in practice that it is generally
beneficial to add the payoff function itself to the set of basis functions.
This may cause problems for the pseudo regression, as the inner products of
the payoff function with the other (typically polynomial) basis functions
cannot be expected to be given in closed form, thereby violating
Assumption~\ref{ass:explicit-G}. However, we can compute those scalar products
numerically, by quadrature, quasi Monte Carlo or even standard Monte Carlo, at
negligible extra cost, especially in the setting of
Assumption~\ref{ass:homog-markov}. With some additional work, we can still
achieve orthonormality by Gram-Schmidt.
\end{remark}

Let us summarize the findings of this section by looking at the most typical
case. Arguably, this is the case when $M \gg K, \mathcal{J}$. We may always
choose basis functions to be orthonormal, hence we consider the second case in
Proposition~\ref{prop:cost-pseudo-inhom}. For standard regression, the
computational costs are, hence, asymptotically proportional to $\mathcal{J}%
MK^{2}$, whereas the pseudo regression only incurs costs proportional to
$\mathcal{J}MK$. This will lead to a computational advantage, especially when
$K$ is large.

\subsection{Longstaff--Schwartz algorithm}

\label{sec:longst-schw-algor}

Asymptotically, the Longstaff--Schwartz algorithm based on standard regression
usually incurs the same cost as the Tsitsiklis--van Roy algorithm based on
standard regression (Proposition~\ref{prop:cost-standard-regression}).

\begin{proposition}
\label{prop:standard-longstaff-schwartz} The computational cost of the
Longstaff--Schwartz algorithm due to (\ref{SLS})--(\ref{SLS2}), using standard
regression, is
\[
\mathcal{C}_{reg} = \mathcal{O}\left(  \mathcal{J} M K^{2} + \mathcal{J} K^{3}
\right)  .
\]

\end{proposition}

\begin{proof}
First we simulate all trajectories at cost $\mathcal{O}(\mathcal{J}M)$ and
evaluate the basis functions along all simulated values at cost $\mathcal{O}%
(\mathcal{J}MK)$. For each step $j$ in the backward iteration we need to set
up the matrix $\mathcal{N}^{\top}\mathcal{N}$ at (individual) cost
$\mathcal{O}(MK^{2})$. Then we need to compute the right hand side
$\mathcal{Y}^{(j)}$ at cost $\mathcal{O}(\mathcal{J}M)$, which assumes that
the values of the continuation function at times $j+1, \ldots, \mathcal{J}$
have been pre-computed in the earlier stages of the backward iteration.
Finally, compute the coefficients at cost $\mathcal{O}(MK^{2} + K^{3})$.
\end{proof}

If we apply the Longstaff--Schwartz algorithm with pseudo-regression we note
an important difference compared to Tsitsiklis--van Roy: in the standard case
of the algorithm (presented in Section~\ref{sec:pseudo-regr-LS}) we
potentially have to evaluate the basis functions $\psi_{1}, \ldots, \psi_{K}$
for each sample $X^{j-1, \mathcal{U}^{m}}_{r}$, $r = j, \ldots, \mathcal{J}$.
In the worst case, this incurs costs proportional to $\mathcal{J}^{2} K M$ on
top. Hence, we obtain

\begin{proposition}
\label{prop:pseudo-LS} The computational cost of the Longstaff--Schwartz
algorithm based on pseudo regression is
\[
\mathcal{C}_{pseudo} = \mathcal{O}\left(  \mathcal{J}K^{2} + K^{3} +
\mathcal{J}^{2} MK \right)  .
\]
If the basis functions are orthonormal w.r.t.~$\mu$, then the costs reduce to
\[
\mathcal{C}_{pseudo} = \mathcal{O}\left(  \mathcal{J}^{2} MK \right)  .
\]

\end{proposition}

Suppose that we are actually in the setting of
Assumption~\ref{ass:homog-markov}. Then we may once again duplicate the
samples. In this case, we still need to simulate full trajectories starting
from the sampled initial points at $j=0$, but we can then ``shift'' those
samples in time. In this case, we only need to evaluate the basis functions at
$X^{0, \mathcal{U}^{m}}_{r}$, $r = j, \ldots, \mathcal{J}$, which incurs an
additional cost $\mathcal{O}(\mathcal{J}MK)$. On the other hand, we will get a
cost component $\mathcal{O}(\mathcal{J}^{2}M)$ simply from assembling
$\mathcal{Y}^{(j)}$ for each $j$. In total, we obtain

\begin{proposition}
\label{prop:pseudo-LS-copy} If Assumption~\ref{ass:homog-markov} holds and we
duplicate samples, then the computational cost of the Longstaff--Schwartz
algorithm based on pseudo regression is
\[
\mathcal{C}_{pseudo} = \mathcal{O}\left(  \mathcal{J}K^{2} + K^{3} +
\mathcal{J} MK + \mathcal{J}^{2} M \right)  .
\]
If the basis functions are orthonormal w.r.t.~$\mu$, then the costs reduce to
\[
\mathcal{C}_{pseudo} = \mathcal{O}\left(  \mathcal{J} MK + \mathcal{J}^{2} M
\right)  .
\]

\end{proposition}

Let us, once again, summarize the discussion on the computational costs by
looking at a typical case. For true Bermudan options, $\mathcal{J}$ is, of
course, fixed, while $M$ and $K$ need to be increased in order to improve the
accuracy of the estimator. Hence, the typical case for asymptotic
considerations is probably $M \gg K \gg\mathcal{J}$. Once again, we may very
well assume to have chosen orthonormal basis functions together with
Assumption~\ref{ass:homog-markov}. Hence, regarding pseudo regression, we are
in the second case of Proposition~\ref{prop:pseudo-LS-copy}. Under these
conditions, the computational cost of the Longstaff--Schwartz algorithm with
standard regression is asymptotically proportional to $\mathcal{J}MK^{2}$,
while the pseudo regression incurs cost asymptotically proportional to
$\mathcal{J} MK$. Again, the costs of the standard regression dominate in the
long run.

\section{Numerical experiments}

\label{sec:numer-exper}

The numerical experiments below are run on a laptop computer with an
Intel\textregistered\, Core\texttrademark\,i7-6500U processor and 8GB RAM. All
algorithms are implemented and executed in GNU Octave version 4.0.3 running on
openSUSE Leap 42.3. {Moreover, all the codes are single-threaded.}

We consider a Bermudan Max-Call option on $n$ assets which, for instance, has
already been considered in \cite{andersonbroadie}. The assets $X^{i}$ are
identically distributed and yield dividends with rate $\delta$. They are given
as the solutions to
\begin{align}
\label{sde}dX^{i}_{t} = (r - \delta) X^{i}_{t} dt + \sigma X^{i}_{t}
dW^{i}_{t},\quad X^{i}_{0}=x_{0}, \quad t\in[0, T], \quad i=1, \ldots, n,
\end{align}
where $W^{i}$ are independent scalar Brownian motions. The interest rate $r$
as well as $\sigma$ are constant. We assume to have $\mathcal{J}+1$ exercise
dates $0\leq t_{0}< t_{1}, \ldots< t_{\mathcal{J}}\leq T$ in which the option
holder may exercise to obtain the payoff%
\begin{align}
\label{payofffkt}h(X_{t})=\left(  \max(X_{t}^{1}, \ldots, X^{n}_{t}%
)-\kappa\right)  ^{+},
\end{align}
where $\kappa>0$. Moreover, we introduce the discounted payoff function by
$f_{t}(X_{t}) = \operatorname{e}^{-r t}h(X_{t})$.\smallskip

Throughout the remainder of this section, it is assumed that $T=3$, $r=0.05$,
$\delta= 0.1$, $\sigma=0.2$ and $\kappa=100$. We further choose $t_{j}=
j\frac{T}{\mathcal{J}}$ for $j=0, \ldots, \mathcal{J}$.

\subsection{Option pricing using Tsitsiklis--van Roy}

\label{numericsTvR} We aim to determine an approximation $\bar v_{0}$ of the
value $Y_{0}$ of the Bermudan Max-Call option above. To do so, we use the
algorithm of Tsitsiklis--van Roy \cite{TV2001}, where the computation of the
continuation functions $\bar c_{j}$ ($j=0, \ldots, \mathcal{J}-1$) is based on
the standard regression (SR). On the other hand we use the algorithm explained
in Section \ref{alg} (see also Algorithm~\ref{alg:pseudo-bermudan}), in which
SR is replaced by a pseudo regression (PR) method.

\smallskip

In this section, we set $\mathcal{J}=9$. The general idea within the PR based
algorithm is to choose random initial values $\mathcal{U}_{t_{j}}\sim
\mu_{t_{j}}$ at every exercise date $t_{j}$ for each component $X^{i}$ that
is  given through (\ref{sde}). In fact, we obtain good results for this scheme
by setting
\begin{align*}
\mathcal{U}_{t_{j}}=\operatorname{e}^{\operatorname{m}_{t_{j}} + \hat
\sigma_{t_{j}} Z},
\end{align*}
where $Z\sim\mathcal{N}(0, 1)$, i.e., $\mathcal{U}_{t_{j}}$ is log-normal
distributed, a constant variance parameter $\hat\sigma_{t_{j}}\equiv\hat
\sigma\in\left[ \sigma\sqrt{T/2}, \sigma\sqrt{T}\right] $ and a mean parameter
$\operatorname{m}_{t_{j}}= (r-\delta)t_{j}-0.5 \hat\sigma^{2}+\ln(x_{0})$
ensuring $\mathbb{E }[\mathcal{U}_{t_{j}}]=\mathbb{E }[X^{i}_{t_{j}}]$.
However, it turns out that the PR algorithm is not very sensitive in the mean
parameter such that we can choose a constant one, i.e., $\operatorname{m}%
_{t_{j}}\equiv\operatorname{m}= (r-\delta)t-0.5 \sigma^{2} t+\ln(x_{0})$ for a
fixed $t\in[T/2, T]$. This means that we choose $\mathcal{U}_{t_{j}}%
\equiv\mathcal{U}$ (or $\mu_{t_{j}}\equiv\mu$) independent of the exercise
date which has the advantage of reducing the number of basis function
evaluations to two. Hence, the PR schemes is computationally even cheaper than
before.\smallskip

The particular choice of the mean parameter $\operatorname{m}$ and a variance
parameter $\hat\sigma$ depends on $n$ since we observe that we obtain somewhat
better results if these parameter are slightly modified with the number of
assets $n$.\smallskip

We choose orthonormal polynomials $(\psi_{k})_{k=1, \ldots, K}$ with respect
to $\mu$. To be more precise, we introduce Hermite polynomials on $\mathbb{R}$
of degree $i$ which we denote by $H_{i}$. We then define $\psi_{1}, \ldots,
\psi_{K}$ via a suitable ordering of the functions
\begin{align*}
\prod_{j=1}^{n} H_{i_{j}}\left(  \frac{\ln(y_{j})-\operatorname{m}}{\hat
\sigma}\right)
\end{align*}
with $i_{1}+ i_{2}+ \ldots+ i_{n} \leq p$, where $p\in\mathbb{N}$ is the
largest polynomial degree and $i_{j}\in\mathbb{N}$ for all $j$. Thus, the
total number of basis functions is $K=\frac{(p+n)!}{p! \,n!}$. In fact, we
just pick all the products of Hermite polynomials with total degree up to $p,$
and use some bijection to assign them to an index $k=1, \ldots, K$. In the
following, the above orthonormal functions are not just used for the PR but
also within the SR ansatz.\smallskip

We now determine $\bar v_{0}$ of the Bermudan option for the initial values
$x_{0}= 90, 100, 110$. Moreover, we conduct the numerical experiments for
$n=2, 3, 4, 5$ assets. We set $p=4$ for $n=5$, else we fix $p=5$. The reason
for using a different polynomial degree for $n=5$ is that we aim to {achieve
the same accuracy} for both the PR and the SR based algorithm in order to be
able to compare both schemes. With this choice of $p$ the same output is
approximately obtained by using  $M=2$e$+06$ samples for both methods in the
derivation of the continuation functions. \smallskip

Of course, one can also choose a different number of samples for PR than for
SR in order to find the exactly the same output for both cases. However, it is
hard to find these numbers of samples such that both algorithms yield exactly
the same output $\bar v_{0}$. Notice that the algorithm with a slightly lower
number for $\bar v_{0}$ can always be improved by using more
samples.\smallskip

We start with the case $n=2$ for which we fix $\operatorname{m}= (r-\delta-0.5
\sigma^{2}) \frac{T}{2}+\ln(x_{0})$ and $\hat\sigma=0.26$. We see that both PR
and SR perform almost equally well for the case $n=2$, see first block of
Table \ref{tablev0}. It seems that PR even yields slightly better results for
$\bar v_{0}$. The respective computational times can be found in Figure
\ref{n2TvR}. It turns out that the PR algorithm is more than three times
faster.\smallskip

It is possible to also take the same mean and variance parameter for $n>2$.
However, enlarging the variance slightly to $\hat\sigma=0.29$ for $n=3$ leads
to results which are a little bit better. The mean parameter remains the same
as for $n=2$. From the second block in Table \ref{tablev0}, it can be seen
that both algorithms lead to approximately the same value $\bar v_{0}$. The
advantage of using PR is the much lower computational time. We know from
Figure \ref{n3TvR} that we save a factor of more than five compared to
SR.\smallskip

In the case of $n=4$, the variance parameter is again enlarged compared to the
case of having $n=3$ assets, i.e., $\hat\sigma=0.32$. We modify the mean
parameter as well which in this case is $\operatorname{m}= (r-\delta-0.5
\sigma^{2}) 2.56+\ln(x_{0})$. Again, SR and PR yield results of the same
quality, see third block of Table \ref{tablev0}, whereas the PR based
algorithm is extremely fast compared to the SR. Figure \ref{n4TvR} shows that
a factor of more than nine can be saved.\smallskip

From $n=4$ to $n=5$, only the mean parameter is changed to $\operatorname{m}=
(r-\delta-0.5 \sigma^{2}) 3+\ln(x_{0})$. SR and PR provide approximately the
same outputs, see fourth block of Table \ref{tablev0}. Moreover, Figure
\ref{cpu_time_TvRn5} shows a more than seven times faster computation of $\bar
v_{0}$ using PR. \begin{table}[th]
\centering
\begin{tabular}
[c]{cccc}\hline
$x_{0}$ & $\bar v_{0}(x_{0})$ based on PR & $\bar v_{0}(x_{0})$ based on SR &
Parameter used\\\hline
$n=2$ &  &  & \\
$90$ & $8.046\; (0.006)$ & $8.030\; (0.006)$ & $K=21$, $\hat\sigma=0.26$\\
$100$ & $13.884\; (0.008)$ & $13.868\; (0.008)$ & and\\
$110$ & $21.322\; (0.009)$ & $21.314\; (0.009)$ & $\operatorname{m}=\ln
(x_{0})-0.105$\\
$n=3$ &  &  & \\
$90$ & $11.238\;(0.007)$ & $11.234 \;(0.007)$ & $K=56$, $\hat\sigma=0.29$\\
$100$ & $18.640\;(0.009)$ & $18.640\;(0.009)$ & and\\
$110$ & $27.533\;(0.010)$ & $27.520\;(0.010)$ & $\operatorname{m}=\ln
(x_{0})-0.105$\\
$n=4$ &  &  & \\
$90$ & $14.045\;(0.008)$ & $14.049\;(0.008)$ & $K=126$, $\hat\sigma=0.32$\\
$100$ & $22.638\;(0.009)$ & $22.638\;(0.010)$ & and\\
$110$ & $32.527\;(0.011)$ & $32.531\;(0.011)$ & $\operatorname{m}=\ln
(x_{0})-0.179$\\
$n=5$ &  &  & \\
$90$ & $16.567\;(0.008)$ & $16.562 \;(0.008)$ & $K=126$, $\hat\sigma=0.32$\\
$100$ & $26.054\;(0.010)$ & $26.055\;(0.010)$ & and\\
$110$ & $36.665\;(0.011)$ & $36.657\;(0.011)$ & $\operatorname{m}=\ln
(x_{0})-0.21$\\
&  &  & \\
&  &  & \\
&  &  &
\end{tabular}
\caption{Approximative value $\bar v_{0}$ of the Bermudan option based on SR
and PR for $\mathcal{J}=9$ using Tsitsiklis--van Roy. The computation of the
continuation functions is based on $M=2$e$+06$ samples.}%
\label{tablev0}%
\end{table}

\begin{figure}[th]
\centering
\begin{subfigure}[b]{1\textwidth}
\includegraphics{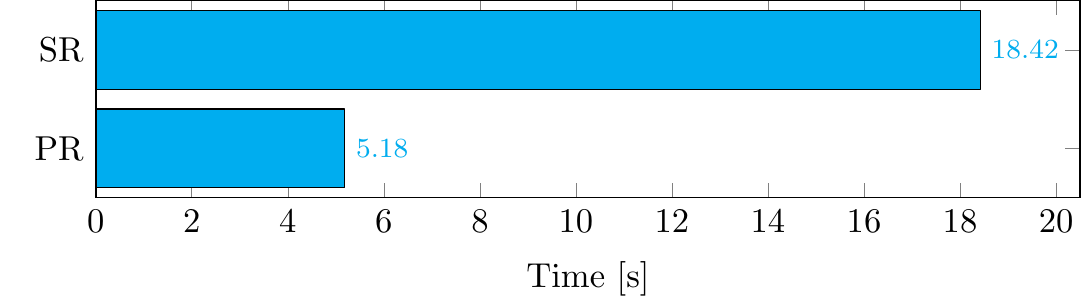}
\caption{Computational times for $n=2$.}\label{n2TvR}
\end{subfigure}\vspace{0.25cm} \begin{subfigure}[b]{1\textwidth}
\includegraphics{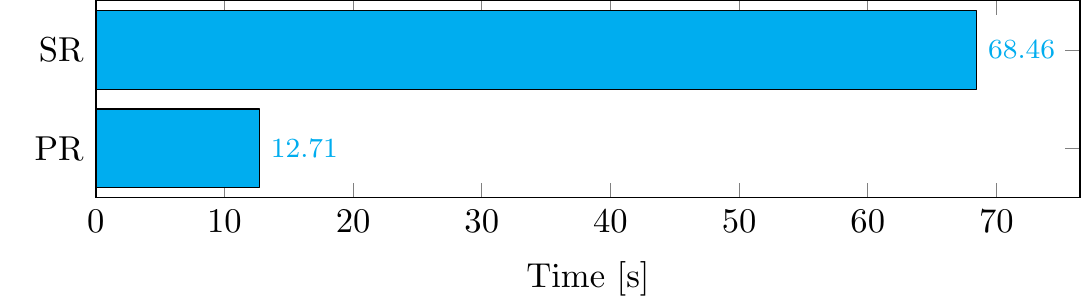}
\caption{Computational times for $n=3$.}\label{n3TvR}
\end{subfigure}\vspace{0.25cm} \begin{subfigure}[b]{1\textwidth}
\includegraphics{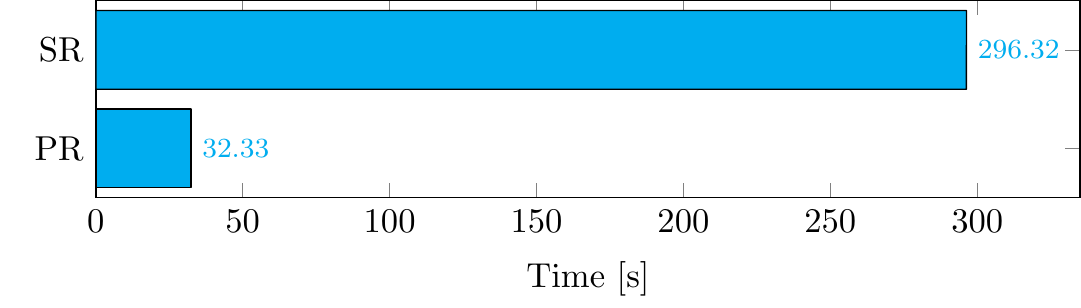}
\caption{Computational times for $n=4$.}\label{n4TvR}
\end{subfigure}
\begin{subfigure}[b]{1\textwidth}
\includegraphics{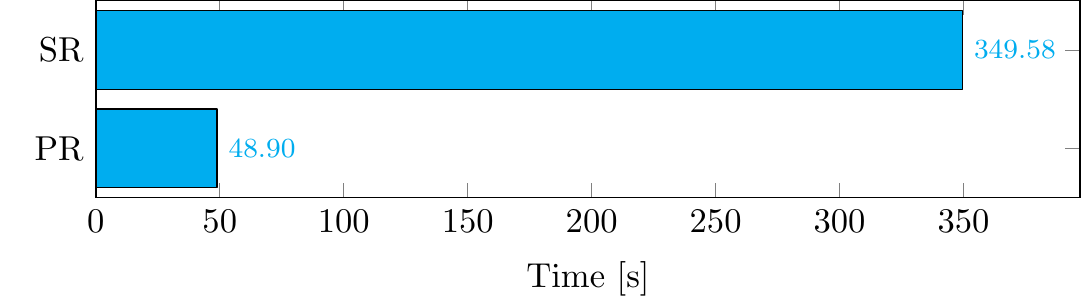}
\caption{Computational times for $n=5$.}\label{cpu_time_TvRn5}
\end{subfigure}
\caption{Time to compute continuation functions using Tsitsiklis--van Roy and
$M=2$e$+06$ samples; standard versus pseudo regression for $n=2, 3, 4, 5$ and
$\mathcal{J}=9$.}%
\label{cpu_time_TvR}%
\end{figure}\smallskip

\subsection{Option pricing using Longstaff--Schwartz}

In this section, we study the same problem as in Section \ref{numericsTvR}.
However, we now determine the approximation $\bar v_{0}$ for the value of the
Bermudan Max-Call option based on Longstaff--Schwartz. Again, we investigate a
PR based version (see Section \ref{sec:pseudo-regr-LS}) and a SR based method
(see \cite{LS2001}) for this approach in order to derive the continuation
functions $\bar c_{j}$.\smallskip

We observe from the simulations that PR is clearly faster than SR for
Longstaff--Schwartz for $n>2$ using the same parameter as in Section
\ref{numericsTvR}. Rather than presenting the same experiments for
Longstaff--Schwartz again, we only point out a typical case that exhibits a 
large gain in computational time:  the case $n=4$ and
$\mathcal{J}=4$, and all other parameter as before. Table \ref{tablev0LS}
shows us that both types of regression provide a similar value for $\bar
v_{0}$ but using SR is more than two times more expensive, compare Figure
\ref{cpu_time_LS}. \begin{table}[th]
\centering
\begin{tabular}
[c]{cccc}\hline
$x_{0}$ & $\bar v_{0}(x_{0})$ based on PR & $\bar v_{0}(x_{0})$ based on SR &
Parameter used\\\hline
$n=4$ &  &  & \\
$90$ & $13.719\;(0.008)$ & $13.708\;(0.008)$ & $K=126$, $\hat\sigma_{n}%
=0.32$\\
$100$ & $22.170\;(0.010)$ & $22.163\;(0.010)$ & and\\
$110$ & $31.914\;(0.011)$ & $31.915\;(0.011)$ & $\operatorname{m}_{n}%
=\ln(x_{0})-0.179$\\
&  &  & \\
&  &  &
\end{tabular}
\caption{Approximative value $\bar v_{0}$ of the Bermudan option based on SR
and PR for $\mathcal{J}=4$ using Longstaff--Schwartz.}%
\label{tablev0LS}%
\end{table}

\begin{figure}[th]
\centering
\includegraphics{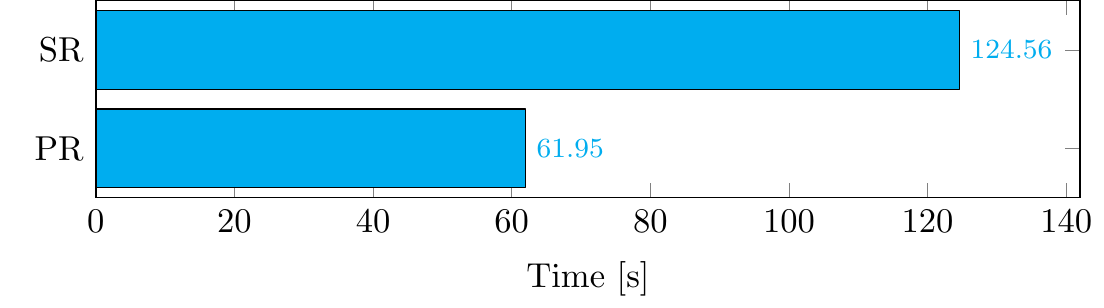} \caption{Time to compute continuation functions
using Longstaff--Schwartz; standard versus pseudo regression for $n=4$ and
$\mathcal{J}=4$.}%
\label{cpu_time_LS}%
\end{figure}

\begin{remark}
The gain of the pseudo-regression version of the Longstaff--Schwartz algorithm
with respect to the standard one is generally smaller than the gain in Section
\ref{numericsTvR}, in the context of Tsitsiklis--van Roy. The reason is clear:
The backward steps (\ref{SLS})--(\ref{SLS2}) in the standard algorithm cannot
be straightforwardly modified to a pseudo-regression based setting. Therefore
in order to construct a vector $\mathcal{Y}^{(j)}$ of stopped cash-flows in a
pseudo-regression based setup analogue to the standard method, new
trajectories are simulated for each exercise date starting from an initial
state simulated under $\mu.$ Of course, this makes the procedure more costly,
however, the generally expensive inversion of a random matrix at each exercise
date is avoided this way.
\end{remark}

\section{Conclusions}

\label{sec:conclusions}

We compare the classical regression method for solving Bermudan options in the
form of either the Tsitsiklis--van Roy or the Longstaff--Schwartz algorithm
with a new variant based on pseudo regression, i.e., Monte Carlo simulation of
$L^{2}$ inner products based on samples simulated from an artificial base
measure $\mu$ (not directly related to the distribution of the underlying
stock price at any point in time).

As a key issue, setting up and inverting a random matrix at every exercise
date, can be avoided in the pseudo-regression approach. Therefore, the
computational cost can be considerably lower for a similar level of accuracy,
as verified in numerical examples in Section~\ref{sec:numer-exper}. This is
also motivated from an asymptotic cost analysis, see
Section~\ref{sec:computational-cost}. Furthermore, the pseudo-regression based
algorithms are much easier to analyze theoretically, and leads to shorter and
clearer convergence proofs (see Section~\ref{gvps}). At the same time,
convergence rates are slightly improved as the logarithmic error terms seen in
classical regression error analysis (due to inversion of a random matrix) can
be omitted.

The choice of the probability measure $\mu$ turns out to be crucial for the
success of the pseudo regression algorithm. On the other hand, the procedure
is insensitive enough w.r.t.~$\mu$ that one single measure can in many cases
be chosen for each time step. However, $\mu$ has to be chosen in an
appropriate way since bad choices can lead to vastly increased errors of the
ultimately computed option prices. As a rule, $\mu$ should be chosen such that
$\mathcal{U}\sim\mu$ covers the important areas of the underlying stochastic
process $X_{j}$ for all $j$. These areas can usually be estimated roughly from
the dynamics of the underlying process.

In some sense, the particular choice of $\mu$ can be compared to importance
sampling, as it allows to change the distribution of the actual sampled
points, without inducing bias. As such we expect that the flexibility in the
choice of $\mu$ could be advantageous in particular in situations like deep
out-of-the-money options, when the payoff is positive only on a rare event. We
will study these aspects in more details in future work.

\appendix

\section{Proofs}

\subsection{Proof of Theorem~\ref{psth}}

\label{sec:proof-theorem-psth}

Let $u^{K}$ be the projection of $u$ on to the linear span of $\psi
_{1},...,\psi_{K},$ i.e.,%
\begin{equation}
u^{K}=\underset{w\,\in\,\spn \{\psi_{1},...,\psi_{K}\}}{\arg\inf}%
\int_{\mathcal{D}}\left\vert w(z)-u(z)\right\vert ^{2}\mu(dz). \label{pr}%
\end{equation}
Then, with $\gamma:=(\gamma_{1},...,\gamma_{K})^{\top}\in\mathbb{R}^{K}$
defined by%
\begin{equation}
u^{K}=\sum_{k=1}^{K}\gamma_{k}\psi_{k}, \label{uK}%
\end{equation}
and $\alpha\in\mathbb{R}^{K}$ defined by $\alpha_{k}\coloneqq
\left\langle \psi_{k} \, , u \right\rangle _{L^{2}(\mu)},$ it follows
straightforwardly by taking scalar products that
\begin{equation}
\gamma=\mathcal{G}^{-1}\alpha. \label{ag}%
\end{equation}
By the rule of Pythagoras it follows that,%
\begin{gather}
\mathbb{E}\int_{\mathcal{D}}\left\vert \overline{u}(z)-u(z)\right\vert ^{2}%
\mu(dz)=\label{th11}\\
\mathbb{E}\int_{\mathcal{D}}\left\vert \overline{u}(z))-u^{K}(z)\right\vert
^{2}\mu(dz)+\int_{\mathcal{D}}\left\vert u^{K}(z)-u(z)\right\vert ^{2}%
\mu(dz).\nonumber
\end{gather}
Hence, the second term in (\ref{th1}) is clear due to (\ref{pr}) and
(\ref{th11}). With $\psi(z):=\left(  \psi_{1}(z),...,\psi_{K}(z)\right)
^{\top}$ we obtain for the first term in (\ref{th11}) that
\begin{align*}
&  \mathbb{E}\int_{\mathcal{D}}\left\vert \overline{u}(z)-u^{K}(z)\right\vert
^{2}\mu(dz)=\int_{\mathcal{D}}\mathbb{E}\left\vert \overline{\beta}^{\top}%
\psi(z)-\gamma^{\top}\psi(z)\right\vert ^{2}\mu(dz)\\
&  =\int_{\mathcal{D}}\mathbb{E}\left\vert \left(  \frac{1}{M}\mathcal{Y}%
^{\top}\mathcal{M}-\alpha^{\top}\right)  \mathcal{G}^{-1}\psi(z)\right\vert
^{2}\mu(dz)\\
&  =\int_{\mathcal{D}}\mathbb{E}\left[  \left(  \frac{1}{M}\mathcal{Y}^{\top
}\mathcal{M}-\alpha^{\top}\right)  \mathcal{G}^{-1}\psi(z)\psi^{\top
}(z)\mathcal{G}^{-1}\left(  \frac{1}{M}\mathcal{M}^{\top}\mathcal{Y}%
-\alpha\right)  \right]  \mu(dz)\\
&  =\mathbb{E}\left[  \left(  \frac{1}{M}\mathcal{Y}^{\top}\mathcal{M}%
-\alpha^{\top}\right)  \mathcal{G}^{-1}\left(  \frac{1}{M}\mathcal{M}^{\top
}\mathcal{Y}-\alpha\right)  \right]  ,
\end{align*}
using (\ref{bt}), (\ref{ut}), (\ref{uK}), (\ref{ag}), and
\[
\int_{\mathcal{D}}\left[  \psi(z) \psi^{\top}(z)\right]  _{kl}\mu
(dz)=\langle\psi_{k},\psi_{l}\rangle=\mathcal{G}_{kl}.
\]
We thus have%
\begin{align*}
0  &  \leq\mathbb{E}\int_{\mathcal{D}}\left\vert \overline{u}(z)-u^{K}%
(z)\right\vert ^{2}\mu(dz)\\
&  \leq\frac{1}{\underline{\lambda_{\min}}}\mathbb{E}\left\vert \frac{1}%
{M}\mathcal{N}^{\top}\mathcal{Y}-\alpha\right\vert ^{2}=\frac{1}%
{\underline{\lambda_{\min}}}\sum_{k=1}^{K}\operatorname{Var}\left[  \frac
{1}{M}\mathcal{M}^{\top}\mathcal{Y}\right]  _{k},
\end{align*}
since%
\begin{align}
\mathbb{E}\left[  \frac{1}{M}\mathcal{M}^{\top}\mathcal{Y}\right]  _{k}  &
=\frac{1}{M}\mathbb{E}\sum_{m=1}^{M}\psi_{k}(\mathcal{U}^{(m)})Y^{(m)}%
\label{Krem}\\
&  =\mathbb{E}\left(  \psi_{k}(\mathcal{U}^{(1)})Y^{(1)}\right)
=\mathbb{E}\left(  \psi_{k}(\mathcal{U}^{(1)})\mathbb{E}\left[  Y^{(1)}%
\,|\,\mathcal{U}^{(1)}\right]  \right) \nonumber\\
&  =\langle\psi_{k},u\rangle=\alpha_{k}.\nonumber
\end{align}
Now, by observing that
\begin{align}
\operatorname{Var}\left[  \frac{1}{M}\mathcal{M}^{\top}\mathcal{Y}\right]
_{k}  &  =\operatorname{Var}\left(  \frac{1}{M}\sum_{m=1}^{M}\psi
_{k}(\,\mathcal{U}^{(m)})Y^{(m)}\right) \label{Krem1}\\
&  =\frac{1}{M}\operatorname{Var}\left(  \psi_{k}(\,\mathcal{U}^{(1)}%
)Y^{(1)}\right) \nonumber\\
&  =\frac{1}{M}\mathbb{E}\operatorname{Var}\left[  \psi_{k}(\,\mathcal{U}%
^{(1)})Y^{(1)}|\,\mathcal{U}^{(1)}\right]  +\frac{1}{M}\operatorname{Var}
\mathbb{E}\left[  \psi_{k}(\,\mathcal{U}^{(1)})Y^{(1)}|\,\mathcal{U}%
^{(1)}\right] \nonumber\\
&  =\frac{1}{M}\mathbb{E}\text{\thinspace}\left(  \psi_{k}^{2}(\,\mathcal{U}%
^{(1)})\operatorname{Var}\left[  Y^{(1)}|\,\mathcal{U}^{(1)}\right]  \right)
+\frac{1}{M}\operatorname{Var}\psi_{k}(\,\mathcal{U}^{(1)})u\left(
\,\mathcal{U}^{(1)}\right) \nonumber\\
&  \leq\frac{\sigma^{2}+D^{2}}{M}\mathcal{G}_{kk}^{K},\nonumber
\end{align}
one has%
\begin{align*}
\frac{1}{\underline{\lambda_{\min}}}\sum_{k=1}^{K}\operatorname{Var}\left[
\frac{1}{M}\mathcal{M}^{\top}\mathcal{Y}\right]  _{k}\leq\frac{\sigma
^{2}+D^{2}}{M\underline{\lambda_{\min}}}\text{tr}\left(  \mathcal{G}%
^{K}\right)  \leq\frac{\sigma^{2}+D^{2}}{M\underline{\lambda_{\min}}%
}K\overline{\lambda_{\max}},
\end{align*}
and then (\ref{th1}) follows.

\begin{remark}
{From (\ref{Krem}) we see that we are essentially approximating the inner
products $\left\langle \psi_{k} \, , u \right\rangle _{L^{2}(\mu)}$ by a
simple Monte Carlo simulation. At a first glance one may estimate the squared
error due to (\ref{Krem1}) as being proportional to $K^{2}/M$ (up to the
projection error itself). Thus, Theorem~\ref{psth} states that actually this
error is proportional to $K/M$, even when the basis functions are not
orthogonal.}
\end{remark}

\subsection{Proof of Lemma \ref{lem23}}

\label{sec:proof-lemma-reflem23}

Let $X$ be a generic trajectory independent of $\mathcal{G}_{j+1},$ and let us
represent a family of optimal stopping times $\tau_{j}^{\ast},$
$j=1,...,\mathcal{J},$ by $\tau_{\mathcal{J}}^{\ast}=\mathcal{J},$ and for
$j<\mathcal{J},$
\[
\tau_{j}^{\ast}:=j\,1_{\bigl\{f_{j}(X_{j})\geq c_{j}(X_{j})\bigr\}}+\tau
_{j+1}^{\ast}1_{\bigl\{f_{j}(X_{j})<c_{j}(X_{j})\bigr\}}.
\]
For $j<\mathcal{J}$ we then have,%
\begin{gather*}
f_{\tau_{j+1}^{\ast}}(X_{\tau_{j+1}^{\ast}})-f_{\overline{\tau}_{j+1}%
}(X_{\overline{\tau}_{j+1}})=\left(  f_{j+1}(X_{j+1})-f_{\overline{\tau}%
_{j+1}}(X_{\overline{\tau}_{j+1}})\right)  1_{\{\tau_{j+1}^{\ast
}=j+1,\overline{\tau}_{j+1}>j+1\}}\\
+\left(  f_{\tau_{j+1}^{\ast}}(X_{\tau_{j+1}^{\ast}})-f_{j}(X_{j+1})\right)
1_{\{\tau_{j+1}^{\ast}>j+1,\overline{\tau}_{j+1}=j+1\}}\\
+\left(  f_{\tau_{j+1}^{\ast}}(X_{\tau_{j+1}^{\ast}})-f_{\overline{\tau}%
_{j+1}}(X_{\overline{\tau}_{j+1}})\right)  1_{\{\tau_{j+1}^{\ast
}>j+1,\overline{\tau}_{j+1}>j+1\}}.
\end{gather*}
By denoting temporarily $\mathsf{E:=}\mathbb{E}_{\mathcal{G}_{j+1}},$ and
denoting $\mathcal{R}_{j}:=\mathsf{E}\left[  \left.  f_{\tau_{j+1}^{\ast}%
}(X_{\tau_{j+1}^{\ast}})-f_{\overline{\tau}_{j+1}}(X_{\overline{\tau}_{j+1}%
})\right\vert X_{j}\right]  ,$ we have $\mathcal{R}_{j}\geq0$ almost surely,
and%
\begin{align}
\mathcal{R}_{j}  &  =\mathsf{E}\left[  \left.  \left(  f_{j+1}(X_{j+1}%
)-\mathsf{E}\left[  \left.  f_{\overline{\tau}_{j+2}}(X_{\overline{\tau}%
_{j+2}})\right\vert X_{j+1}\right]  \right)  1_{\{\tau_{j+1}^{\ast
}=j+1,\overline{\tau}_{j+1}>j+1\}}\right\vert X_{j}\right] \nonumber\\
&  +\mathsf{E}\left[  \left.  \left(  \mathsf{E}\left[  \left.  f_{\tau
_{j+2}^{\ast}}(X_{\tau_{j+2}^{\ast}})\right\vert X_{j+1}\right]
-f_{j+1}(X_{j+1})\right)  1_{\{\tau_{j+1}^{\ast}>j+1,\overline{\tau}%
_{j+1}=j+1\}}\right\vert X_{j}\right] \nonumber\\
&  +\mathsf{E}\left[  \left.  \mathsf{E}\left[  \left.  f_{\tau_{j+2}^{\ast}%
}(X_{\tau_{j+2}^{\ast}})-f_{\overline{\tau}_{j+2}}(X_{\overline{\tau}_{j+2}%
})\right\vert X_{j+1}\right]  1_{\{\tau_{j+1}^{\ast}>j+1,\overline{\tau}%
_{j+1}>j+1\}}\right\vert X_{j}\right] \nonumber\\
&  \eqqcolon T_{1}+T_{2}+\mathsf{E}\left[  \left.  \mathcal{R}_{j+1}%
1_{\{\tau_{j+1}^{\ast}>j+1,\overline{\tau}_{j+1}>j+1\}}\right\vert
X_{j}\right]  . \label{c1}%
\end{align}
For $T_{1}$ we have%
\begin{align*}
T_{1}  &  =\mathsf{E}\left[  \left.  \left(  f_{j+1}(X_{j+1})-\mathsf{E}%
\left[  \left.  f_{\tau_{j+2}^{\ast}}(X_{\tau_{j+2}^{\ast}})\right\vert
X_{j+1}\right]  \right)  1_{\{\tau_{j+1}^{\ast}=j+1,\overline{\tau}%
_{j+1}>j+1\}}\right\vert X_{j}\right] \\
&  +\mathsf{E}\left[  \left.  \left(  \mathsf{E}\left[  \left.  f_{\tau
_{j+2}^{\ast}}(X_{\tau_{j+2}^{\ast}})\right\vert X_{j+1}\right]
-\mathsf{E}\left[  \left.  f_{\overline{\tau}_{j+2}}(X_{\overline{\tau}_{j+2}%
})\right\vert X_{j+1}\right]  \right)  1_{\{\tau_{j+1}^{\ast}=j+1,\overline
{\tau}_{j+1}>j+1\}}\right\vert X_{j}\right]  ,
\end{align*}
and since
\begin{align*}
\overline{c}_{j+1}(X_{j+1})  &  \geq f_{j+1}(X_{j+1})\geq\mathsf{E}\left[
\left.  f_{\tau_{j+2}^{\ast}}(X_{\tau_{j+2}^{\ast}})\right\vert X_{j+1}\right]
\\
&  =c_{j+1}(X_{j+1})\geq\mathsf{E}\left[  \left.  f_{\overline{\tau}_{j+2}%
}(X_{\overline{\tau}_{j+2}})\right\vert X_{j+1}\right]
\end{align*}
on $\{\tau_{j+1}^{\ast}=j+1,\overline{\tau}_{j+1}>j+1\},$ we get
\begin{align}
0  &  \leq T_{1}\leq\mathsf{E}\left[  \left.  \left(  \overline{c}%
_{j+1}(X_{j+1})-c_{j+1}(X_{j+1})\right)  1_{\{\tau_{j+1}^{\ast}=j+1,\overline
{\tau}_{j+1}>j+1\}}\right\vert X_{j}\right] \nonumber\\
&  +\mathsf{E}\left[  \left.  \mathcal{R}_{j+1}1_{\{\tau_{j+1}^{\ast
}=j+1,\overline{\tau}_{j+1}>j+1\}}\right\vert X_{j}\right]  . \label{c2}%
\end{align}
Similarly, for $T_{2}$, we find
\begin{equation}
0\leq T_{2}\leq\mathsf{E}\left[  \left.  \left(  c_{j+1}(X_{j+1})-\overline
{c}_{j+1}(X_{j+1})\right)  1_{\{\tau_{j+1}^{\ast}>j+1,\overline{\tau}%
_{j+1}=j+1\}}\right\vert X_{j}\right]  . \label{c3}%
\end{equation}
Combining (\ref{c1}), (\ref{c2}), and (\ref{c3}), yields%
\[
\mathcal{R}_{j}\leq\mathsf{E}\left[  \left.  \left\vert \overline{c}%
_{j+1}(X_{j+1})-c_{j+1}(X_{j+1})\right\vert \right\vert X_{j}\right]
+\mathsf{E}\left[  \left.  \mathcal{R}_{j+1}\right\vert X_{j}\right]  .
\]
By straightforward induction, using the tower property and the final condition
$\mathcal{R}_{\mathcal{J}-1}=0,$ we then obtain%
\[
0\leq c_{j}\left(  X_{j}\right)  -\widetilde{c}_{j}\left(  X_{j}\right)
\leq\sum_{l=j+1}^{\mathcal{J}-1}\mathsf{E}\left[  \left.  \left\vert
\overline{c}_{l}(X_{l})-c_{l}(X_{l})\right\vert \right\vert X_{j}\right]  .
\]
By now taking $X_{l}=X_{l}^{j,\mathcal{U}}$ independent of $\mathcal{G}%
_{j+1},$ and then on both sides the $L_{p}$-norm due to the distribution of
$X_{j}^{j,\mathcal{U}}\sim\mu,$ applying the triangle inequality, and by using
that%
\begin{equation}
\mathsf{E}\left[  \mathsf{E}\left[  \left.  \left\vert \overline{c}_{l}%
(X_{l})-c_{l}(X_{l})\right\vert \right\vert X_{j}\right]  ^{p}\right]
\leq\mathsf{E}\left[  \left\vert \overline{c}_{l}(X_{l})-c_{l}(X_{l}%
)\right\vert ^{p}\right]  , \label{help}%
\end{equation}
we finally obtain (\ref{eq:bound_1}).

\subsection{Proof of Lemma \ref{lem23TV}}

\label{sec:proof-lemma-reflem23TV}

Let $X$ be a generic trajectory independent of $\mathcal{G}_{j+1}.$ Then for
$j<\mathcal{J}$ (see (\ref{cont1})),%
\begin{multline*}
\left\vert c_{j}\left(  X_{j}\right)  -\widetilde{c}_{j}\left(  X_{j}\right)
\right\vert \leq\left\vert \mathsf{E}\left[  \left.  \max\left[
f_{j+1}\left(  X_{j+1}\right)  ,c_{j+1}\left(  X_{j+1}\right)  \right]
-\max\left[  f_{j+1}\left(  X_{j+1}\right)  ,\overline{c}_{j+1}\left(
X_{j+1}\right)  \right]  \right\vert X_{j}\right]  \right\vert \\
\leq\mathsf{E}\left[  \left.  \left\vert c_{j+1}\left(  X_{j+1}\right)
-\overline{c}_{j+1}\left(  X_{j+1}\right)  \right\vert \right\vert
X_{j}\right]  .
\end{multline*}
By now taking $X_{j+1}=X_{j+1}^{j,\mathcal{U}}$ independent of $\mathcal{G}%
_{j+1},$ on both sides the $L_{p}$-norm due to the distribution of
$X_{j}^{j,\mathcal{U}}\sim\mu,$ and using (\ref{help}), we get (\ref{boundTV}).

\subsection{Proof of Theorem~\ref{thm_main}}

\label{sec:proof-theorem-refeq}

The theorem will be proved by induction. Due to Theorem \ref{psth} we have
(note that $\overline{c}_{j}$ and $\widetilde{c}_{j}$ are random functions)
almost surely that
\begin{align}
\mathbb{E}_{\mathcal{G}_{j+1}}\left[  \left\Vert \overline{c}_{j}%
-\widetilde{c}_{j}\right\Vert _{L_{2}(\mu_{j})}^{2}\right]   &  \leq C_{1}%
^{2}\frac{K}{M}+C_{2}^{2}\underset{w\,\in\,\text{\textsf{span}}\{\psi
_{1},...,\psi_{K}\}}{\inf}\left\Vert \widetilde{c}_{j}-w\right\Vert
_{L_{2}(\mu_{j})}^{2},\text{ \ \ hence}\nonumber\\
\left\Vert \overline{c}_{j}-\widetilde{c}_{j}\right\Vert _{L_{2}(\mu
_{j}\otimes\mathbb{P})}  &  \leq C_{1}\sqrt{\frac{K}{M}}+C_{2}\underset
{w\,\in\,\text{\textsf{span}}\{\psi_{1},...,\psi_{K}\}}{\inf}\left\Vert
\widetilde{c}_{j}-w\right\Vert _{L_{2}(\mu_{j}\otimes\mathbb{P})} \label{gy}%
\end{align}
for some $C_{1},C_{2}>0,$ which do not depend on $j,K,$ and $M$ and the
distribution $\mu_{j}.$ We now prove the statement (\ref{eq: main}) for
$\eta:=\max(C_{1},C_{2}).$ Since $\widetilde{c}_{\mathcal{J}-1}=c_{\mathcal{J}%
-1}$ for time $\mathcal{J}-1,$ (\ref{eq: main}) is implied by (\ref{gy}) with
$j=\mathcal{J}-1.$ Suppose the statement is proved for $0<j+1\leq
\mathcal{J}-1.$ Let us write,
\begin{multline}
\underset{w\,\in\,\text{\textsf{span}}\{\psi_{1},...,\psi_{K}\}}{\inf
}\left\Vert \widetilde{c}_{j}-w\right\Vert _{L_{2}(\mu_{j}\otimes\mathbb{P}%
)}\label{htr}\\
\leq\left\Vert \widetilde{c}_{j}-c_{j}\right\Vert _{L_{2}(\mu_{j}%
\otimes\mathbb{P})}+\underset{w\,\in\,\text{\textsf{span}}\{\psi_{1}%
,...,\psi_{K}\}}{\inf}\left\Vert c_{j}-w\right\Vert _{L_{2}(\mu_{j})}.
\end{multline}
By using (\ref{gy}), (\ref{htr}), and the unconditional expectation applied to
Lemma \ref{lem23} with $p=2$ we get
\begin{align}
\left\Vert \overline{c}_{j}-c_{j}\right\Vert _{L_{2}(\mu_{j}\otimes
\mathbb{P})}  &  \leq\left\Vert \overline{c}_{j}-\widetilde{c}_{j}\right\Vert
_{L_{2}(\mu_{j}\otimes\mathbb{P})}+\left\Vert \widetilde{c}_{j}-c_{j}%
\right\Vert _{L_{2}(\mu_{j}\otimes\mathbb{P})}\nonumber\\
&  \leq C_{1}\sqrt{\frac{K}{M}}+C_{2}\underset{w\,\in\,\text{\textsf{span}%
}\{\psi_{1},...,\psi_{K}\}}{\inf}\left\Vert c_{j}-w\right\Vert _{L_{2}(\mu
_{j})}\nonumber\\
&  +(C_{2}+1)\left\Vert \widetilde{c}_{j}-c_{j}\right\Vert _{L_{2}(\mu
_{j}\otimes\mathbb{P})}\nonumber\\
&  \leq\eta\varepsilon_{j,M,K}+\left(  \eta+1\right)  \sum_{l=j+1}%
^{\mathcal{J}-1}\left\Vert \overline{c}_{l}-c_{l}\right\Vert _{L_{2}(\mu
_{j,l}\otimes\mathbb{P})}. \label{1a}%
\end{align}
Next observe that%
\begin{multline}
\left\Vert \overline{c}_{l}-c_{l}\right\Vert _{L_{2}(\mu_{j,l}\otimes
\mathbb{P})}^{2}=\int_{\mathbb{R}^{d}}\mathsf{E}\left[  \left\vert
\overline{c}_{l}(x)-c_{l}(x)\right\vert ^{2}\right]  \frac{\mu_{j,l}(x)}%
{\mu_{l}(x)}\mu_{l}(x)dx\label{obs}\\
\leq\mathcal{R}_{\infty}\left\Vert \overline{c}_{l}-c_{l}\right\Vert
_{L_{2}(\mu_{l}\otimes\mathbb{P})}^{2},
\end{multline}
whence (\ref{1a}) yields,%
\begin{equation}
\left\Vert \overline{c}_{j}-c_{j}\right\Vert _{L_{2}(\mu_{j}\otimes
\mathbb{P})}\leq\eta\varepsilon_{j,M,K}+\mathcal{R}_{\infty}^{1/2}\left(
\eta+1\right)  \sum_{l=j+1}^{\mathcal{J}-1}\left\Vert \overline{c}_{l}%
-c_{l}\right\Vert _{L_{2}(\mu_{l}\otimes\mathbb{P})}. \label{in1}%
\end{equation}
Using the induction hypothesis we then have,%
\begin{multline}
\sum_{l=j+1}^{\mathcal{J}-1}\left\Vert \overline{c}_{l}-c_{l}\right\Vert
_{L_{2}(\mu_{j,l}\otimes\mathbb{P})}\leq\label{in2}\\
\sum_{l=j+1}^{\mathcal{J}-1}\eta\varepsilon_{l,M,K}\left(  1+\mathcal{R}%
_{\infty}^{1/2}\left(  \eta+1\right)  \right)  ^{\mathcal{J}-l-1}\\
\leq\eta\varepsilon_{j,M,K}\frac{\left(  1+\mathcal{R}_{\infty}^{1/2}\left(
\eta+1\right)  \right)  ^{\mathcal{J}-j-1}-1}{\mathcal{R}_{\infty}%
^{1/2}\left(  \eta+1\right)  }.
\end{multline}
By then combining (\ref{in1}) and (\ref{in2}) we get (\ref{eq: main}).

\subsection{Proof of Theorem~\ref{thm_mainTV}}

\label{sec:proof-theorem-refeqTV}

The proof, by induction, is similar to the one of Theorem~\ref{thm_main}. Due
to Theorem \ref{psth} we obtain again (see (\ref{gy}))
\begin{equation}
\left\Vert \overline{c}_{j}-\widetilde{c}_{j}\right\Vert _{L_{2}(\mu
_{j}\otimes\mathbb{P})}\leq C_{1}\sqrt{\frac{K}{M}}+C_{2}\underset
{w\,\in\,\text{\textsf{span}}\{\psi_{1},...,\psi_{K}\}}{\inf}\left\Vert
\widetilde{c}_{j}-w\right\Vert _{L_{2}(\mu_{j}\otimes\mathbb{P})} \label{gyTV}%
\end{equation}
for some $C_{1},C_{2}>0,$ which do not depend on $j,K,$ and $M$ and the
distribution $\mu_{j},$ where now $\widetilde{c}_{j}$ is defined by
(\ref{tiltv}). Since $\widetilde{c}_{\mathcal{J}-1}=c_{\mathcal{J}-1}$ for
time $\mathcal{J}-1,$ (\ref{eq: mainTV}) follows from (\ref{gyTV}) with
$\eta=\max(C_{1},C_{2})$ for $j=\mathcal{J}-1.$ Suppose (\ref{eq: mainTV}) is
proved for $0<j+1\leq\mathcal{J}-1.$ Now (\ref{htr}) applies in the present
setting also. Then by (\ref{gyTV}), (\ref{htr}), and the unconditional
expectation applied to Lemma \ref{lem23TV} with $p=2,$ we get analogue to
(\ref{1a})
\begin{equation}
\left\Vert \overline{c}_{j}-c_{j}\right\Vert _{L_{2}(\mu_{j}\otimes
\mathbb{P})}\leq\eta\varepsilon_{j,M,K}+\left(  \eta+1\right)  \left\Vert
\overline{c}_{j+1}-c_{j+1}\right\Vert _{L_{2}(\mu_{j,j+1}\otimes\mathbb{P})}.
\label{1aTV}%
\end{equation}
Next observe that, analogue to (\ref{obs}),
\[
\left\Vert \overline{c}_{j+1}-c_{j+1}\right\Vert _{L_{2}(\mu_{j,j+1}%
\otimes\mathbb{P})}^{2}\leq\mathcal{R}_{+}\left\Vert \overline{c}%
_{j+1}-c_{j+1}\right\Vert _{L_{2}(\mu_{j}\otimes\mathbb{P})}^{2}.
\]
Thus by (\ref{1aTV}) and the induction hypothesis,%
\begin{align*}
\left\Vert \overline{c}_{j}-c_{j}\right\Vert _{L_{2}(\mu_{j}\otimes
\mathbb{P})}  &  \leq\eta\varepsilon_{j,M,K}+\mathcal{R}_{+1}^{1/2}\left(
\eta+1\right)  \left\Vert \overline{c}_{j+1}-c_{j+1}\right\Vert _{L_{2}%
(\mu_{j+1}\otimes\mathbb{P})}\\
&  \leq\eta\varepsilon_{j,M,K}+\mathcal{R}_{\infty}^{1/2}\left(
\eta+1\right)  \eta\varepsilon_{j+1,M,K}\frac{\left(  \mathcal{R}_{+}%
^{1/2}\left(  \eta+1\right)  \right)  ^{\mathcal{J}-j-1}-1}{\mathcal{R}%
_{+}^{1/2}\left(  \eta+1\right)  -1}\\
&  \leq\eta\varepsilon_{j,M,K}\frac{\left(  \mathcal{R}_{+}^{1/2}\left(
\eta+1\right)  \right)  ^{\mathcal{J}-j}-1}{\mathcal{R}_{+}^{1/2}\left(
\eta+1\right)  -1}.
\end{align*}

\bibliographystyle{plain}

\end{document}